\documentclass[onefignum,onetabnum]{siamart220329}

\usepackage[utf8]{inputenc}

\usepackage{amsmath}
\usepackage{amssymb}
\usepackage{graphicx}
\usepackage{pifont}
\usepackage{braket}
\usepackage{etoolbox}
\usepackage{upgreek}
\usepackage{comment}
\usepackage{stmaryrd}
\usepackage{appendix}
\usepackage{dsfont}
\usepackage{url}
\usepackage[shortlabels]{enumitem}
\usepackage{xcolor}
\usepackage{multicol}
\usepackage[normalem]{ulem}
\usepackage{rotating}
\usepackage{makecell}

\makeatletter
\patchcmd{\@addmarginpar}{\ifodd\c@page}{\ifodd\c@page\@tempcnta\m@ne}{}{}
\makeatother
\reversemarginpar

\renewcommand{\marginpar}[1]{}

\newtheorem{remark}[theorem]{Remark}

\newcommand{\generate}[2]{\langle #1 \rangle_{\mathsf{#2}}}
\newcommand{\mc}[1]{\mathcal{#1}}
\newcommand{\mf}[1]{\mathfrak{#1}}
\newcommand{\mbb}[1]{\mathbb{#1}}

\newcommand{\C}{\mathbb{C}}
\newcommand{\R}{\mathbb{R}}
\newcommand{\N}{\mathbb{N}}

\newcommand{\cV}{\mathcal{V}}

\newcommand{\bG}{\mathbf{G}}
\newcommand{\bK}{\mathbf{K}}

\newcommand{\iu}{\mathrm{i}\mkern1mu}

\newcommand{\down}{\shortdownarrow}

\newcommand{\diag}{{\rm diag}}
\newcommand{\relint}{{\rm relint}}

\newcommand{\supp}{{\rm supp}}

\newcommand{\conv}{\operatorname{conv}}
\newcommand{\linspan}{{\rm span}}
\newcommand{\SU}{{\rm SU}}
\newcommand{\GL}{{\rm GL}}
\newcommand{\U}{{\rm U}}
\newcommand{\tr}{{\rm tr}}

\newcommand{\spec}{{\rm spec}}

\newcommand{\Lat}{{\rm Lat}}
\newcommand{\Ad}{{\rm Ad}}
\newcommand{\ad}{{\rm ad}}

\newcommand{\stab}{\mathsf{stab}}
\newcommand{\reach}{\mathsf{reach}}

\newcommand{\derv}{\mathsf{derv}}

\newcommand{\todo}[1]{
\textcolor{blue}{TODO: #1}
}

\renewcommand{\epsilon}{\varepsilon}
\newcommand{\e}{\mathbf e}
\newcommand{\wkl}{{\mf w_{\sf{GKSL}}}}

\headers{Markovian Quantum Systems with Unitary Control}{E.~Malvetti, F.~vom Ende, G.~Dirr, and T.~Schulte-Herbr\"uggen}

\title{Reachability, Coolability, and Stabilizability of Open Markovian Quantum Systems with Fast Unitary Control}

\author{
Emanuel Malvetti\thanks{School of Natural Sciences, Technische Universit\"at M\"unchen, 85737 Garching, Germany, and Munich Centre for Quantum Science and Technology (MCQST) \& Munich Quantum Valley (MQV)} 
\and 
Frederik vom Ende\thanks{Dahlem Center for Complex Quantum Systems, Freie Universit{\"a}t Berlin, 14195 Berlin, Germany} 
\and
Gunther Dirr\thanks{Department of Mathematics, University of W{\"u}rzburg, 97074 W{\"u}rzburg, Germany}
\and 
\mbox{Thomas Schulte-Herbr\"uggen\footnotemark[1]}}

\begin{document}


\maketitle

\begin{abstract}
Open Markovian quantum systems with fast and full Hamiltonian control can be reduced to an equivalent control system on the standard simplex modelling the dynamics of the eigenvalues of the density matrix describing the quantum state. 
We explore this reduced control system for answering questions on reachability and stabilizability with immediate applications to the cooling of Markovian quantum systems. 
We show that for certain tasks of interest, the control Hamiltonian can be chosen time-independent.
--- The reduction picture is an example of dissipative interconversion 
between equivalence classes of states, where the classes are induced by fast controls.
\end{abstract}



\section{Introduction}

Often a major obstacle towards realizing quantum technologies
roots in uncontrolled or unmitigated noise.
Hence systematic effort is being put into achieving significant progress in reducing noise in current hardware 
(see, e.g., the quantum technology roadmap~\cite{Roadmap2018} and refs.\ therein) on one hand. 
On the other hand, quantum optimal control~\cite{dAless21,DiHeGAMM08}
lends itself to complement these efforts  
to further mitigate noise on the `software' side, 
or in other cases to modulate noise in
order to even exploit it as additional control resource beyond coherent controls (see, e.g., the quantum control roadmap~\cite{Koch22} and refs.\ therein). 
A practical instance is quantum error correction with noise-assisted quantum feedback~\cite{Rouchon19}.
In any case, every quantum system that can be externally controlled must interact with its environment and hence is also subject to decoherence.
Thus we accept noise as natural `part of the game'
when studying what can be achieved in spite of---or even thanks to---such noise.

Moreover, in this work we assume the noise to be Markovian and time-independent in the sense that it is described by a master equation of \textsc{gks}--Lindblad form~\cite{GKS76,Lindblad76}. 
Furthermore, in the systems of concern, we assume that unitary control is fast compared to dissipation.
Corresponding results can be obtained assuming that the noise itself is switchable as in the experimental set-up of~\cite{Mart14,McDermott_TunDissip_2019}.
We address several fundamental control-theoretic questions.
For example: Which states and subspaces can be stabilized? Which states can be reached? 
Is the system coolable, controllable, or accessible? 
We emphasize that certain stabilization and reachability tasks
of interest---although assuming fast control over the entire unitary group for our
analysis---can be implemented by time-independent control Hamiltonians. Indeed,
some results in this direction already exist in the literature~\cite{Kraus08,TV09,Schirmer10}, and our results
show what improvements can (and cannot) be obtained from using time-dependent
Hamiltonian control. A particular section is devoted to the unital case where our
general answers can be further specialized.


The main tool used here is a reduction of a full bilinear control system~\cite{Elliott09,Jurdjevic97} 
$\dot X(t) = (A+\sum_ju_jB_j) X(t)$ evolving on the space of Hermitian matrices $X(t)$,
to a reduced one $\dot \lambda(t) = -L_U \lambda(t)$ describing the dynamics of the eigenvalues of $X(t)$.
The reduced system is obtained by factoring out the unitary action, which is possible as soon as one has fast unitary controllability.
In the finite-dimensional quantum dynamical systems treated here, henceforth $X(t)\equiv\rho(t)$ is a density matrix representing the state of the system, and so its eigenvalues sum up to $1$.
Thus $\lambda(t)$ lives in the standard simplex, which then forms the reduced state space.
Obviously, such a reduced system is easier to analyze (and visualize) than the full set of density matrices. 
See Figure~\ref{fig:sketch} for an illustration of this approach.

\begin{figure}[ht]
\centering
\includegraphics[width=0.94\textwidth]{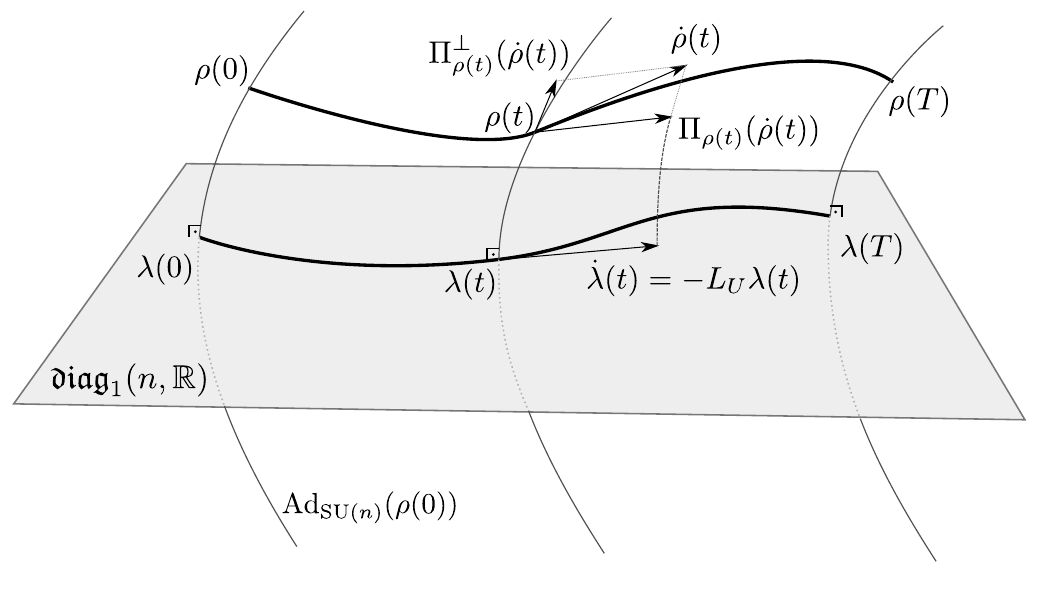}
\caption{Relationship between the time evolutions of a bilinear control system~\eqref{eq:bilinear-control-system} 
on density matrices $\rho(t)$ and the reduced control system~\eqref{eq:simplex-control-system} governing the dynamics of the eigenvalues of $\rho(t)$, where ``the'' vector $\lambda(t)$ of eigenvalues is depicted by the respective diagonal matrix $\diag(\lambda(t))$, see Section~\ref{sec:tools}.
The derivative $\dot\rho(t)$ can always be split into a part orthogonal to the orbit (using the orthogonal projection $\Pi_{\rho(t)}$ onto the commutant of $\rho(t)$), and a part tangent to the orbit (using the complementary projection $\Pi_{\rho(t)}^\perp$). 
We depict only the regular case, where $\Pi_{\rho(t)}=\Ad_{U(t)}\circ\Pi_{\diag}\circ\Ad_{U(t)}^{-1}$ for $\rho(t)=U(t)\diag(\lambda(t))U(t)^*$.
A central result in this work is the Equivalence Theorem~\ref{thm:equivalence} in the main text which details the equivalence of the two control systems.
}
\label{fig:sketch}
\end{figure}

The general idea of reducing control systems admitting fast control on a Lie group action has been addressed, 
e.g., in~\cite[Ch.~22]{Agrachev04}---however under two simplifying assumptions: (i) commuting controls and (ii)
that the reduced state space has no singularities. 
(In this work, the singularities are exactly the collisions of eigenvalues of the density matrix).
In a companion paper~\cite{MDES23}, we have 
generalized this idea in the setting of symmetric Lie algebras without invoking such assumptions.
Indeed the singularities are an inherent feature and present 
the main complication of the reduction.
The reduced control system in the present Lindbladian setting 
has, to the best of our knowledge, 
first been formulated in~\cite{Sklarz04}, in~\cite{Yuan10} for a single qubit, and it has been studied in~\cite{rooney2018}. --- Our treatment in~\cite{MDES23} also 
removes those assumptions made in~\cite{rooney2018} where certain singularities are essentially ignored.
A natural way to simplify the reduced control system is to restrict its controls to values in a finite set (in our case the permutations of the eigenvalues of the state). 
This is also considered in~\cite{rooney2018}, and explored more thoroughly in~\cite{CDC19,OSID23,vE_PhD_2020,MTNS2020_1} 
in the context of quantum thermodynamics: 
there, \/`thermal operations\/' come with separate time evolutions of diagonal and off-diagonal terms in the
density matrix, which naturally inspires the reduction to \/`toy models\/' of diagonal states.
Applications to unital systems were already given in~\cite{Yuan10,Styliaris19} and the results will be recovered here as special cases.

Note that in the context of quantum control theory similar ideas have been explored:
if the quotient 
space 
$\bG/\bK$ (where $\bG$ is the entire group of transformations and $\bK$ is the group brought about by fast controls)
is a Riemannian symmetric space, strong results on time-optimal control have been derived~\cite{Khaneja01b}.
Unfortunately, beyond two-qubit systems such scenarios are rare in practice~\cite{Khaneja01b,Khaneja02,Khaneja08}:
For general $n$-body systems with $n>2$ the time-optimal control problem in the setting above becomes a
hard \textit{sub-Riemannian} geodesic problem~\cite{Khaneja02}
(as the coupling Hamiltonians are limited to the usual two-body interactions).\footnote{%
Since sub-Riemannian geometry is notoriously intricate~\cite{BelRis96,Mo02,Agrachev19},  
%
meaningful limits to minimal control times have been given under more special conditions, e.g., a single control~\cite[Prop.~4.6]{chambrion2009controllability}.
For an overview on the relation of time-optimal control to \/`quantum-speed limits\/' see~\cite{deffnerJPA2017}.}
Since quotient spaces $\bG/\bK$ are rather complicated, often
one contents oneself with finding their diameters
to derive general speed limits~\cite{Gauthier21,Burgarth23}.

\section*{Outline and Main Results}

The main tools used in this paper are presented in Section~\ref{sec:tools}.
We start with the \textsc{gks}--Lindblad equation and introduce two associated matrix algebras, whose invariant subspaces are important in the structure theory of Kossakowski--Lindblad generators as well as the study of time-independent Hamiltonian control.
We then introduce the 
reduced control system which is of central importance.
It turns out that the aforementioned algebras occur repeatedly in the study of the reduced control system.

When proceeding to study the control-theoretic properties of our open Markovian
quantum system, we use (i) the reduced control system to find algebraic characterizations and (ii) the equivalence principle to lift the results to the full control system.
In Section~\ref{sec:stab} we establish
stabilizability of individual points and the entire system (Thm.~\ref{thm:stab-system}), 
the viability of faces of the simplex, and the accessibility of the system
(Prop.~\ref{prop:accessibility}, and for unital systems Prop.~\ref{prop:unital-dir-access}).
Section~\ref{sec:reach} is devoted to reachability, in particular (asymptotic) coolability (Thm.~\ref{thm:asymptotic-coolability}), reverse coolability, and the reachability of faces, with the conditions for (approximate) controllability of the system settled by Prop.~\ref{prop:approx-ctrl}.
Finally, Section~\ref{sec:unital} treats
the special case of unital quantum systems, where the algebraic structure simplifies considerably and allows to derive stronger results 
on reachability (Thm.~\ref{thm:reach-unital}). 
A compact overview is provided in Table~\ref{tab:conclusion}.

\section{Tools and Methods} \label{sec:tools}

The following definition encapsulates what we mean by an open Markovian quantum system with fast and full Hamiltonian control.
Let $\{H_j\}_{j=1}^m$ be a set of Hermitian matrices, called \emph{control Hamiltonians}, and $I$ an interval of the form $[0,T]$ or $[0,\infty)$.
A path $\rho:I\to\mf{pos}_1(n)$ of density matrices is a solution of the bilinear control system
\begin{align} \label{eq:bilinear-control-system} \tag{\sf B}
\dot\rho(t) = 
-\Big(\iu \sum_{j=1}^m u_j(t) \ad_{H_j}+L\Big)(\rho(t)), 
\quad \rho(0)=\rho_0\in\mf{pos}_1(n)
\end{align}
with locally integrable control functions $u_j:I\to\R$ if $\rho$ is absolutely continuous and satisfies~\eqref{eq:bilinear-control-system} almost everywhere. 
We will always assume that the control Hamiltonians generate at least the special unitary Lie algebra: $\generate{\iu H_j:j=1,\ldots,m}{Lie}\supseteq\mf{su}(n)$.

\subsection{Kossakowski--Lindblad Generators and their Relaxation Algebras} \label{sec:lindblad}

Throughout this work, we use $n$-dimensional Hilbert spaces ($2\leq n<\infty$) represented by $\C^n$.
A \emph{state} $\rho$ is a density matrix, i.e.~a positive semi-definite operator with unit trace. 
The set of all states is denoted $\mf{pos}_1(n)$.
The Markovian evolution of a state is described by the \emph{\textsc{gks}--Lindblad equation}~\cite{GKS76,Lindblad76}, which has the form
\begin{align} \label{eq:lindblad-equation}
\dot\rho = -L(\rho) = -\Big(\iu\ad_{H_0} +\sum_{k=1}^r\Gamma_{V_k}\Big)(\rho),
\end{align}
where $\ad_{H_0}(\rho):=[H_0,\rho]$ and $\Gamma_{V_k}(\rho) := \frac{1}{2}(V_k^*V_k\rho+\rho V_k^*V_k)- V_k\rho V_k^*$.
The Hamiltonian $H_0\in\iu\mf{u}(n)$ is a Hermitian matrix and the Lindblad terms $\{V_k\}_{k=1}^r\subset\C^{n\times n}$ are arbitrary matrices. 
We call $-L$ a \emph{Kossakowski--Lindblad generator}\footnote{The signs are chosen such that the real parts of the eigenvalues of $-L$ are non-positive.}, and we denote the set of all Kossakowski--Lindblad generators in $n$-dimensions by $\wkl(n)$, called the \emph{Kossakowski--Lindblad Lie wedge}, cf.~\cite{DHKS08}.

Throughout the paper we assume that $-L\in\wkl(n)$ is a Kossakowski--Lindblad generator on the $n$-dimensional Hilbert space $\C^n$, given by Lindblad terms $\{V_k\}_{k=1}^r$ and Hamiltonian $H_0$.
We say that a given set $\{V'_k\}_{k=1}^s\subset\mathbb C^{n\times n}$ is a choice of Lindblad terms of $-L$ if there exists a Hermitian $H_0'\in\mathbb C^{n\times n}$ such that $-L=-(\iu\ad_{H_0'}+\sum_{k=1}^s\Gamma_{V_k'})$.
The freedom of representation of $-L$ is summarized in Lemma~\ref{lemma:freedom-of-reps}.
Recall that $-L$ is called \emph{unital} if $L(\mathds1)=0$ and \emph{purely Hamiltonian} if $-L=-\iu\ad_{H_0}$. 

We will briefly introduce the main concepts here, referring to Appendix~\ref{app:lindblad-eq} for precise statements.
To the generator $-L$ we define an associated (complex) matrix algebra $\cV$, called its \emph{relaxation algebra}, generated by ``the'' Lindblad terms $\{V_k\}_{k=1}^r$ and the identity matrix, as well as its \emph{extended relaxation algebra} $\cV^+$ additionally generated by $K=\iu H_0 + \frac12\sum_{k=1}^rV_k^*V_k$.
Importantly, these matrix algebras are well-defined and their invariant subspaces encode important information about the generator $-L$. 
Indeed, many results about the structure of Kossakowski--Lindblad generators presented in~\cite{BNT08b} can be formulated succinctly using the algebras $\cV$ and $\cV^+$.
Moreover, $\cV$ was used to give a sufficient condition for the evolution to be relaxing in the sense of having a unique attractive fixed point, cf.~\cite{Davies70,Spohn77}.
The invariant subspaces of $\cV$ are called \emph{lazy subspaces}, and invariant subspaces of $\cV^+$ are called \emph{collecting subspaces}. Under the evolution of $-L$, a state supported on a lazy subspace will only leave the subspace ``slowly'', whereas a state supported on a collecting subspace will not leave the subspace at all. Importantly, using time-independent controls, any lazy subspace can be turned into a collecting one.
A collecting subspace whose orthocomplement is also collecting is called an \emph{enclosure}, and it corresponds to a symmetry of the generator $-L$, cf.~\cite{BNT08b,Albert14}.
Interestingly, in the unital case, $\cV$ and $\cV^+$ turn out to be $*$-algebras, which are highly structured.
This allows us to derive strong results, beyond the known fact that a unital system is relaxing if and only if $\cV=\{V_k:1,\ldots,r\}''=\C^{n\times n}$ (with $\{\cdot\}''$ denoting the double-commutant in $\C^{n\times n}$),
which is similar to the result in~\cite{Spohn77}.


These relaxation algebras are useful in understanding what can be achieved using time-independent Hamiltonian control, cf.~\cite{Kraus08,TV09,Schirmer10} as well as Appendix~\ref{app:lindblad-eq}, and they also turn out to be crucial in the study of the reduced control system, which we will introduce in the following section.




\subsection{Reduced Control System}

Our main focus is the reduced control system obtained from the full system~\eqref{eq:bilinear-control-system} using the fast controllability on the unitary group.
The definitions in this section specialize those of~\cite[Sec.~2]{MDES23}, and the results of~\cite{MDES23} then establish equivalence of these systems in a certain sense.

Since the bilinear control system~\eqref{eq:bilinear-control-system} allows for unbounded control functions and since the control Hamiltonians generate the entire special unitary Lie algebra---meaning that we have fast unitary control---we can quickly move within the unitary orbits. Thus we may concentrate on the dynamics of the eigenvalues of the state (as two density matrices have the same spectrum if and only if they lie on the same unitary orbit).\footnote{As mentioned in the introduction, if the noise term is switchable, then one can effectively emulate~\eqref{eq:bilinear-control-system} even if the unitary control is not fast, at the expense of working on slower time scales.}

The reduced state space will be the standard simplex
\begin{align*}
\Delta^{n-1}=\Big\{(x_1,\ldots,x_n)\in\R^n : \sum_{i=1}^nx_i=1,\, x_i\geq0\,\; \forall i\Big\}\,,
\end{align*}
representing the subset of diagonal density matrices.
The standard simplex $\Delta^{n-1}$ is a convex polytope of dimension $n-1$ and its faces are lower dimensional simplices.
An important group action on the simplex, stemming from the action of $\SU(n)$ on $\mf{pos}_1(n)$, is that of the symmetric group $S_n$ acting by coordinate permutations.
Indeed every unitary orbit in $\mf{pos}_1(n)$ intersects $\Delta^{n-1}$ a finite number of times and the intersections form a permutation group orbit (this is just the unitary diagonalization of Hermitian matrices).
Two faces of the same dimension can always be mapped to each other and can thus be considered equivalent.
The points in the (relative) interior of a $(d-1)$-dimensional face correspond to quantum states of rank $d$.
In particular the vertices $e_1,\ldots,e_n$ correspond to the pure states and the barycenter $\e/n$ where, here and henceforth, $\e:=(1,\ldots,1)^\top$ corresponds to the maximally mixed state $\mathds1/n$.
Moreover we use the notation $\Delta^{n-1}_\down := \{x\in\Delta^{n-1} : x_1\geq\ldots\geq x_n\}$, which we also call the \emph{ordered Weyl chamber}, as well as $\mathrm{spec}^\down:\mf{pos}_1(n)\to\Delta^{n-1}_\down$ for the map which arranges the eigenvalues of the input into a vector in non-increasing order.
Conversely, $\diag:\Delta^{n-1}\to\mf{pos}_1(n)$ maps a vector to the corresponding diagonal matrix.
Finally, for $\lambda\in\Delta^{n-1}$ we write $\lambda^\down$ for the non-increasingly ordered version in $\Delta^{n-1}_\down$.

The next step is to define the appropriate control system on the standard simplex.
To motivate the definition, consider any solution $\rho:I\to\mf{pos}_1(n)$ to the bilinear control system~\eqref{eq:bilinear-control-system} and assume that $\rho$ is regular\footnote{A state $\rho$ is called \emph{regular} if its eigenvalues are all distinct.\label{footnote_regular}} on $I$.
Moreover let $\rho=\Ad_U(\diag(\lambda)):=U\diag(\lambda)U^*$ be a differentiable (in time) diagonalization of $\rho$.
Then by differentiating (cf.~\cite[Lem.~2.3]{Malvetti22-diag}) one can show that $\dot\lambda=-L_U\lambda$, where
\begin{align} \label{eq:LU}
-L_U 
:= 
-\Pi_{\diag}\circ\Ad_U^{-1}\circ L\circ \Ad_U\circ \diag(\cdot)\,,
\end{align}
and where $\Pi_{\diag}:\mf{pos}_1(n)\to\Delta^{n-1}$ maps a matrix to the vector of its diagonal elements.
The $-L_U$ are clearly linear and we call them \emph{induced vector fields} on $\mathbb R^n$.
Note that, by definition, the $L_U$ are independent of the choice of Lindblad terms $V_k$.
We sketched the situation in Figure~\ref{fig:sketch}.
A more explicit form of the induced vector fields is given by
\begin{equation}\label{eq:def_JU}
-L_U=J(U)-\diag(J(U)^\top\e), \qquad J(U):=\sum_{k=1}^r U^*V_kU\circ\overline{U^*V_kU}
\end{equation}
with $\circ$ the Hadamard product, i.e.~$\langle i|J(U)|j\rangle=\sum_{k=1}^r|\langle i|U^*V_kU|j\rangle|^2$ for all $i,j$.
Indeed this follows from the following computation:
\begin{align} 
-(L_U)_{ij} 
&= \langle i| (\Ad_U^{-1}\circ  L \circ \Ad_U)(|j\rangle\langle j|)  |i\rangle \notag\\
&= \sum_{k=1}^r |(U^*V_kU)_{ij}|^2 - \delta_{ij} \sum_{k=1}^r (U^*V_k^*V_kU)_{ii}\notag
\\&= \sum_{k=1}^r |(U^*V_kU)_{ij}|^2 - \delta_{ij} \sum_{k=1}^r \sum_{\ell=1}^n|(U^*V_kU)_{\ell i}|^2
= J_{ij}(U) - \delta_{ij}\sum_{\ell=1}^nJ_{\ell i}(U)\,.\label{eq:J-comp}
\end{align}

We denote the set of induced vector fields as
$$
\mf L:=\{-L_U:U\in\SU(n)\}\,.
$$
Note that $\mf L$ is the image of a compact set under the continuous function $U\mapsto -L_U$, hence compact itself.
Also, the elements of $\mf L$ are generators of stochastic matrices, i.e.~$\e^\top L_U=0$ and the off-diagonal elements are non-negative.
We write this as $\mf L\subseteq\mf{stoch}(n)\,$, where $\mf{stoch}(n)$ denotes the Lie wedge corresponding to $\mathrm{Stoch}(n)$ which is the closed subsemigroup of $\GL(n,\R)$ consisting of all invertible stochastic matrices.
Now $\mf L\subseteq\mf{stoch}(n)$ in particular means that the standard simplex $\Delta^{n-1}$ is (forward) invariant under the flow of the induced vector fields $-L_U$.

We define on $\Delta^{n-1}$ the set-valued function $\derv$ of \emph{achievable derivatives} by
\begin{align*}
\derv(\lambda)
:= 
\{-L_U\lambda : U\in\SU(n)\} = \mf L\lambda \subset {\sf T}_\lambda\Delta^{n-1}
\end{align*}
where ${\sf T}_\lambda\Delta^{n-1}$ denotes the tangent cone at $\lambda$, which can always be identified with a subset of 
$\R^n_0:=\{x\in\R^n:x_1+\cdots+x_n=0\}$.

With this we are ready to define the reduced control system (in two equivalent ways):

\begin{definition}
A function $\lambda:I\to\Delta^{n-1}$ is a solution of the control system
\begin{align}
\label{eq:simplex-control-system}\tag{\sf R}
\dot\lambda(t) = -L_{U(t)} \lambda(t)\,,
\quad \lambda(0)=\lambda_0\in\Delta^{n-1}
\end{align}
with measurable control function $U:I\to\SU(n)$, if $\lambda$ is absolutely continuous and satisfies \eqref{eq:simplex-control-system} almost everywhere. 
Equivalently\footnote{
This is due to Filippov's Theorem, cf.~\cite[Thm.~2.3]{Smirnov02}. Here by ``equivalent'' we mean that the two systems have exactly the same set of solutions.
},
a solution $\lambda:I\to\Delta^{n-1}$ is an absolutely continuous function which satisfies the differential inclusion
\begin{align*}
\dot\lambda(t)\in\derv(\lambda(t)), \quad \lambda(0)=\lambda_0\in\Delta^{n-1}
\end{align*}
almost everywhere.
\end{definition}

A convenient relaxation of the reduced control system can be obtained by allowing convex combinations of achievable derivatives.

\begin{definition}
A function $\lambda:I\to\Delta^{n-1}$ is a solution of the control system
\begin{align}
\label{eq:relaxed-control-system} \tag{\sf C}
\dot\lambda(t) \in \conv(\derv(\lambda(t))), 
\quad \lambda(0)=\lambda_0,
\end{align}
if it is absolutely continuous and satisfies the differential inclusion~\eqref{eq:relaxed-control-system} almost everywhere.
\end{definition}
\noindent 
The relaxation to the convex hull will slightly enlarge the set of solutions; however, every solution of~\eqref{eq:relaxed-control-system} can still be approximated uniformly (on compact time intervals) by solutions to~\eqref{eq:simplex-control-system}, see~\cite[Ch.~2.4, Thm.~2]{Aubin84}.

\begin{remark} \label{rmk:reduced-aliter}
Again due to Filippov's Theorem, a solution $\lambda$ to~\eqref{eq:simplex-control-system} (resp.~\eqref{eq:relaxed-control-system}) is an absolutely continuous function for which there exists $(-M_\tau)_{\tau\in I}\subset \mf L$ (resp.~$\operatorname{conv}\mf L$) measurable such that $\lambda(t)=\lambda_0+\int_0^t (-M_\tau) \lambda(\tau) \,d\tau$ holds for all $t\in I$.
\end{remark}


The main results of~\cite{MDES23}, namely Theorems~3.8 and 3.14, pertain to this setting as follows.

\begin{theorem}[Equivalence Theorem] \label{thm:equivalence}
Let $\rho:[0,T]\to\mf{pos}_1(n)$ be a solution to the bilinear control system~\eqref{eq:bilinear-control-system} and let $\lambda^\down:[0,T]\to\Delta^{n-1}_\down$ be the unique path which satisfies $\lambda^\down=\mathrm{spec}^\down(\rho)$.
Then $\lambda^\down$ is a solution to the reduced control system~\eqref{eq:simplex-control-system}.
Conversely, let $\lambda:[0,T]\to\Delta^{n-1}$ be a solution to the reduced control system~\eqref{eq:simplex-control-system} with control function $U:[0,T]\to\SU(n)$. Then for every $\epsilon>0$ there exists a solution $\rho_{\varepsilon}:[0,T]\to\mf{pos}_1(n)$ to the bilinear control system~\eqref{eq:bilinear-control-system} such that $$\|\Ad_U(\diag(\lambda))-\rho_{\varepsilon}\|_{\infty}\leq\epsilon\,.$$
\end{theorem}
The proof is given in Appendix~\ref{app:equivalence}. \smallskip


\begin{corollary}
\label{coro:ham-control-system}
Let $H:[0,T]\to\iu\mf{su}(n)$ be an integrable Hamiltonian. Consider a solution $\rho:[0,T]\to\mf{pos}_1(n)$ to $\dot\rho(t)=-(\iu\ad_{H(t)}+L)\rho(t)$. 
Then $\spec^\down(\rho(t))$ is a solution to \eqref{eq:simplex-control-system} and a fortiori to~\eqref{eq:relaxed-control-system}.
\end{corollary}

\begin{proof}
Choose any control Hamiltonians $\{H_j\}_{j=1}^m$ which linearly span $\iu\mf{su}(n)$.
Then we may write $H(t)=\sum_{j=1}^m u_j(t) H_j$ with integrable control functions $u_j$. 
Hence $\rho$ can be seen as a solution to~\eqref{eq:bilinear-control-system} (with the chosen controls) and so we may apply Theorem~\ref{thm:equivalence} to obtain that $\spec^\down(\rho(t))$ is a solution of~\eqref{eq:simplex-control-system}.
\end{proof}
The proof above uses the fact that while ``the'' bilinear control system~\eqref{eq:bilinear-control-system} requires a choice of control Hamiltonians to be fully defined, the reduced control system~\eqref{eq:simplex-control-system} is independent of this choice so long as the control Hamiltonians generate the entire special unitary Lie algebra.

We emphasize that although the ``lifting'' part in Theorem~\ref{thm:equivalence} is only approximate in general, in many relevant cases one can obtain stronger results. When the path is regular
(recall footnote~\ref{footnote_regular}),
then the lift can be performed exactly and explicitly using~\cite[Prop.~3.10]{MDES23}. Moreover, due to the structure theory of Kossakowski--Lindblad generators, it is often the case that one can choose time-independent Hamiltonians to achieve practical tasks. This is summarized in Appendix~\ref{app:lindblad-eq}. 

\subsection{Operator Lifts and Lie Wedges} 

Although not used in this paper, we briefly discuss the operator lifts of the full bilinear system~\eqref{eq:bilinear-control-system} and of the reduced system~\eqref{eq:simplex-control-system}.
The definitions, in a more general context, are recalled in~\cite[Sec.~2.2 \& 2.3]{MDES23}.
Considering the operator lifts immediately leads to the study of Lie semigroups and Lie wedges, cf.~\cite{HHL89,Lawson99,DHKS08}.
In particular, the reachable sets of the operator lifts are determined by the Lie saturate $\generate{\cdot}{LS}$ of their respective generators, see~\cite[Sec.~6]{Lawson99}. The generators of~\eqref{eq:bilinear-control-system} are given by $\Omega=\{X+\iu\ad_{H_j}:j=1,\ldots,m\}$, and the generators of~\eqref{eq:simplex-control-system} are exactly the induced vector fields $\mf L$.

\begin{proposition}
The Lie saturate of the full control system~\eqref{eq:bilinear-control-system} is given by $\generate{\Omega}{LS}=\ad_{\mf{su}(n)}\oplus\generate{\Ad_U^{-1}\circ L\circ \Ad_U:U\in\SU(n)}{wedge}$.
\end{proposition}

\begin{proof}
If $-L=\iu\ad_H$ for some $\iu H\in\mf{su}(n)$, then it is clear that the Lie wedge is $\ad_{\mf{su}(n)}$ since it
is a compact Lie algebra.
Hence we can focus on the case where $-L\notin\ad_{\mf{su}(n)}$. Since $\wkl(n)$ is a global Lie wedge (cf.~\cite[Thm.~3.3]{DHKS08}) and since the edge of $\wkl(n)$ equals $\ad_{\mf{su}(n)}$ (e.g.~by Lemma~\ref{lemma:purely-Hamiltonian}~\ref{it:pure-ham-eig}), the result follows from~\cite[Prop.~2.1~(iii)]{MDES23} (which itself is based on~\cite[Prop.~1.37]{HN12}).
\end{proof}

Furthermore, the Lie saturate of the reduced control system~\eqref{eq:simplex-control-system} is $\generate{\mf L}{LS}=\generate{\mf L}{wedge}$.
This follows immediately from the fact that $\mf{stoch}(n)$ is a pointed (hence global) Lie wedge by use of~\cite[Prop.~1.37]{HN12}.
As a consequence we obtain that the Lie saturates are related by $\Pi_{\diag}\circ\generate{\Omega}{LS}\circ\diag=\generate{\mf L}{LS}$, see also~\cite[Lem.~2.3]{MDES23}.

\section{Stabilizability, Viability and Accessibility} \label{sec:stab}

An important task in control theory is that of keeping the state in a certain region of the state space, called viability, or close to some desired state, called stabilizability. 
In this section we characterize viability and stabilizability in the reduced control system and deduce the implications for the full bilinear system.
Moreover we study accessibility in the reduced system and show that non-unital systems are generically directly accessible. 

\subsection{Stabilizable and Strongly Stabilizable Points}

We begin with stabilizable points, emphasizing that our systems do not have feedback, hence why we talk about open-loop stabilizability only.

\begin{definition}
A point $\lambda\in\Delta^{n-1}$ is called \emph{stabilizable} for~\eqref{eq:simplex-control-system} if it holds that $0\in\conv(\derv(\lambda))$. 
The set of all stabilizable points is denoted $\stab_{\ref{eq:simplex-control-system}}$. 
We say that $\lambda$ is \emph{strongly stabilizable} for~\eqref{eq:simplex-control-system} if $0\in\derv(\lambda)$.
\end{definition}
Using previously established notation, $\lambda$ is strongly 
stabilizable if and only if $0\in\mf L\lambda$, and $\lambda$ is stabilizable if and only if $0\in\conv(\mf L\lambda)=(\conv\mf L)\lambda$. \smallskip

It is clear that strong stabilizability implies stabilizability. 
Intuitively, $\lambda$ is stabilizable if any solution to~\eqref{eq:simplex-control-system} starting at $\lambda$ can remain close to $\lambda$ for an arbitrarily long amount of time, and $\lambda$ is strongly stabilizable if the constant path at $\lambda$ is a solution to~\eqref{eq:simplex-control-system}.
For a more precise statement, see~\cite[Sec.~4.3]{MDES23}. 
The following result is a direct specialization of~\cite[Prop.~4.5]{MDES23}.

%

\begin{proposition} \label{prop:ham-stab-is-strongly-stab}
Let $-L\in\wkl(n)$ be any Kossakowski--Lindblad generator, and let $\rho_0\in\mf{pos}_1(n)$, $\lambda_0\in\Delta^{n-1}$ and $U\in\SU(n)$ be given such that $\rho_0=\Ad_U(\diag(\lambda_0))$. Then the following hold: \smallskip
\begin{enumerate}[(i)]
\item If there is some Hamiltonian $H$ such that $-(\iu\ad_H+L)(\rho_0)=0$, then $-L_U\lambda_0=0$ and hence $\lambda_0$ is strongly stabilizable for~\eqref{eq:simplex-control-system}.
\item Conversely, if $\lambda_0$ is regular and strongly stabilizable for~\eqref{eq:simplex-control-system} with $-L_U\lambda_0=0$, then the \emph{compensating Hamiltonian} $H_c:=\ad_{\rho_0}^{-1}\circ\,\Pi_{\rho_0}^\perp\circ L(\rho_0)$ satisfies
\begin{align*}
-(\iu\ad_{H_c} + L)(\rho_0) = 0\,.
\end{align*}
\end{enumerate}
\end{proposition}
Note that the assumption on regularity is generally necessary, as exemplified in~\cite[Ex.~3.12]{MDES23}.

Denoting $\tilde V_k=U^*V_kU$, and analogously $\tilde H_c$ and $\tilde H_0$, the compensating Hamiltonian $H_c$ for a regular, strongly stabilizable state $\rho=\Ad_U(\diag(\lambda))$ can be written more explicitly as follows.
Using basic equivariance properties and the fact that, by regularity of $\rho$, we have $\Pi_\lambda=\Pi_\diag$, we find
$-\iu H_c 
= \Ad_U\circ\ad_\lambda^{-1}\circ\Pi_{\diag}^\perp\circ\Ad_U^{-1}\circ L\circ\Ad_U(\diag(\lambda))$.
Hence
$-\iu\langle i|\tilde H_c+\tilde H_0|j\rangle
=
\sum_{k=1}^r {\langle i|\Gamma_{\tilde V_k}(\lambda)|j\rangle}/{(\lambda_i-\lambda_j)}
$
which expands to
$$
-\iu \langle i|\tilde H_c|j\rangle
=
\iu \langle i|\tilde H_0|j\rangle
+
\sum_{k=1}^r \frac{\frac{\lambda_i+\lambda_j}2 \langle i|\tilde V_k^*\tilde V_k|j\rangle - \sum_{\ell=1}^n\lambda_\ell\langle i|\tilde V_k|\ell\rangle\langle \ell|\tilde V_k^*|j\rangle}{\lambda_i-\lambda_j}\,,
$$
as desired.

Note that more generally, for regular $\rho_0=\Ad_U(\diag(\lambda_0))$ the compensating Hamiltonian $H_c$ exactly cancels out the part of $-L(\rho_0)$ which is tangent to the orbit, cf.~Figure~\ref{fig:sketch}.

Computing the set of all (strongly) stabilizable states is difficult in general. Even in the three-dimensional toy model case considered in~\cite[Sec.~4.2 \& Fig.~5]{OSID23} this is non-trivial, see also~\cite[Sec.~IV]{rooney2018}.

\subsection{Viable Faces}

Viable subsets of the state space are those in which one can remain for an arbitrary amount of time. More precisely they are defined as follows.

\begin{definition}
A subset $S\subseteq\Delta^{n-1}$ is called \emph{viable} for~\eqref{eq:relaxed-control-system} if for every initial point $\lambda_0\in S$, there exists a solution $\lambda:[0,\infty)\to S$ to~\eqref{eq:relaxed-control-system} with $\lambda(0)=\lambda_0$.
\end{definition}
Note that
the singleton set $\{\lambda\}$ is viable for~\eqref{eq:relaxed-control-system} if and only if $\lambda$ is stabilizable for~\eqref{eq:simplex-control-system}, cf.~\cite[Sec.~4.4]{MDES23}.

We are primarily interested in the case where the subset in question is a face of the simplex $\Delta^{n-1}$.
Due to the permutation symmetry of the reduced system
as well as the geometry of the simplex $\Delta^{n-1}$, all $(d-1)$-dimensional faces of the simplex are equivalent.
The main result of this section will relate viable faces of dimension $d-1$ to lazy subspaces (i.e.~common invariant subspaces of the relaxation algebra $\cV$,
or equivalently, of any choice of Lindblad terms $V_k$ which define $-L$,
cf.~Appendix~\ref{app:lindblad-eq}) of dimension $d$.
We begin with a simple lemma:

\begin{lemma}
\label{lemma:tangent-face-lazy-subspace}
Let $F$ be a face of $\Delta^{n-1}$ of dimension $d-1$ with $1\leq d\leq n$. 
Let $\lambda$ be a point in the relative interior\footnote{The relative interior of a set $S$, denoted $\relint(S)$, is the interior of $S$ within its affine hull. Note that for a singleton set $\{s\}$ the relative interior is $\{s\}$ itself.} of $F$ and let $-M\in\conv(\mf L)$ be such that $-M\lambda$ is tangent to $F$. Then there exists a lazy subspace of dimension $d$.
\end{lemma}

\begin{proof}
Using permutations we can assume that $F$ consist of all vectors in $\Delta^{n-1}$ whose last $n-d$ entries are zero. Since $\lambda$ lies in the relative interior or $F$, its first $d$ entries must be strictly positive. By assumption, the last $n-d$ entries of $M\lambda$ must be zero. But this can only be if the block $(d+1,\ldots,n)\times(1,\ldots,d)$ of $M$ is identically $0$. By Lemma~\ref{lemma:LU-zeros} 
there exists $U\in\U(n)$ such that all $U^*V_kU$ are block upper triangular, hence the first $d$ columns of any such $U^*$ span the desired lazy subspace.
\end{proof}

\begin{proposition} \label{prop:viable-faces}
Let $F$ be a face of $\Delta^{n-1}$ of dimension $d-1$ ($1\leq d\leq n$).
Given any Kossakowski--Lindblad generator $-L\in\wkl(n)$, the following are equivalent:\smallskip
\begin{enumerate}[(i)]
\item \label{it:face-stab} $F$ is viable for~\eqref{eq:relaxed-control-system}.
\item \label{it:face-comp-stab} There exists some $-M\in\conv(\mf L)$ whose flow leaves $F$ invariant.
\item \label{it:face-lazy} There exists a lazy subspace of dimension $d$.
\item \label{it:face-fix-point} There exists a stabilizable point for~\eqref{eq:simplex-control-system} in the relative interior of $F$.
\item \label{it:face-fix-point-strong} There exists a strongly stabilizable point for~\eqref{eq:simplex-control-system} in the relative interior of $F$.
\item \label{it:face-tangent} There exists some $-M\in\conv(\mf L)$ and some $\lambda$ in the relative interior of $F$ such that $-M\lambda$ is tangent to $F$.
\item \label{it:face-tangent2} There exists some $U\in\mathrm U(n)$ and some $\lambda$ in the relative interior of $F$ such that $-L_U\lambda$ is tangent to $F$.
\item \label{it:fixed-rho} There exists a stabilizing Hamiltonian $H_S$ and a state $\rho$ of rank $d$ such that $\rho$ is the unique fixed point of $-(\iu\ad_{H_S}+L)$ restricted to the support of $\rho$ (which is automatically a collecting subspace).
\end{enumerate}
\end{proposition}

\begin{proof}
The implications 
\ref{it:face-lazy} $\Rightarrow$ 
\ref{it:face-comp-stab} $\Rightarrow$ 
\ref{it:face-stab} $\Rightarrow$
\ref{it:face-tangent} 
are easy. 
By Lemma~\ref{lemma:tangent-face-lazy-subspace} and its proof it holds that \ref{it:face-tangent} $\Rightarrow$ \ref{it:face-tangent2} $\Rightarrow$ \ref{it:face-lazy}.
That~\ref{it:face-lazy} $\Rightarrow$ \ref{it:fixed-rho} follows from Lemma~\ref{lemma:unique-attractive-fp}.
Finally, \ref{it:fixed-rho} $\Rightarrow$ \ref{it:face-fix-point-strong} follows from Proposition~\ref{prop:ham-stab-is-strongly-stab}, and \ref{it:face-fix-point-strong} $\Rightarrow$ \ref{it:face-fix-point} $\Rightarrow$ \ref{it:face-tangent}
are clear. 
\end{proof}
Recall from Corollary~\ref{coro:stab-ham} that for a lazy subspace $S$, any $H_S$ satisfying 
$P_S^\perp H_S P_S=-P_S^\perp\big(H_0+\frac{1}{2\iu}\sum_{k=1}^rV_k^*V_k\big)P_S$
is a stabilizing Hamiltonian,
where $P_S$ is the orthogonal projection onto $S$.
This is a well-known result, cf.~\cite{Lidar98DFS,BNT08b}, and Proposition~\ref{prop:viable-faces} shows that such time-independent control is always sufficient for making a face viable.

For the vertices $e_i$ of $\Delta^{n-1}$---which correspond to pure states---Proposition~\ref{prop:viable-faces} specializes as follows:

\begin{corollary}
Given any Kossakowski--Lindblad generator $-L\in\wkl(n)$, the following are equivalent:\smallskip
\begin{enumerate}[(i)]
\item Some (equivalently: each) $e_i$ is stabilizable for~\eqref{eq:simplex-control-system}.
\item Some (equivalently: each) $e_i$ is strongly stabilizable for~\eqref{eq:simplex-control-system}.
\item For some (equivalently: each) choice of Lindblad terms $\{V_k\}_{k=1}^r$ of $-L$, the $V_k$ have a common eigenvector. 
\item There is a stabilizing Hamiltonian $H_\psi$ for a pure state\footnote{
Note that this does \textit{not} follow from Proposition~\ref{prop:ham-stab-is-strongly-stab} since $e_i$ is not regular.},
that is, there exists $\|\psi\|=1$ such that
 $-(\iu\ad_{H_\psi}+L)|(|\psi\rangle\langle\psi|)=0$.\smallskip 
\end{enumerate}
Moreover, if $|\psi\rangle$ is a common eigenvector of the $V_k$ with $V_k|\psi\rangle=\lambda_k|\psi\rangle$, then any Hermitian  $H_\psi$ satisfying
$H_\psi\ket\psi = -\big(H_0+(\lambda-\tfrac1{2\iu}\sum_{k=1}^r|\lambda_k|^2)\mathds1+\frac1{2\iu}\sum_{k=1}^r V_k^*V_k\big)\ket\psi$ where $\lambda\in\R$ is arbitrary
is an admissible choice for the stabilizing Hamiltonian. 
\end{corollary}
Later we will see how certain systems can be asymptotically cooled into a pure state.
The corollary above then shows how such a pure state can be stabilized.

\begin{remark}
The permutation symmetry of the reduced control system (cf.~\cite[Lem.~A.2]{MDES23}) yields some further viable subsets.
Indeed every degeneracy plane (i.e.~the fixed point set of some permutation group) is viable. Similarly one can show that the ordered Weyl chamber $\Delta^{n-1}_\down$ as well as each of its permutations are viable.
\end{remark}

\subsection{Stabilizable Systems}

If every point in $\Delta^{n-1}$ is stabilizable for~\eqref{eq:simplex-control-system}, we say that the system is
\emph{stabilizable} for~\eqref{eq:simplex-control-system}. It turns out that stabilizable systems can be characterized in simple algebraic terms.

\begin{theorem} \label{thm:stab-system} 
Given any $-L\in\wkl(n)$, the following are equivalent: \smallskip
\begin{enumerate}[(i)]
\item 
The system is stabilizable for~\eqref{eq:simplex-control-system}.
\item 
For some (equivalently: each) choice of Lindblad terms $\{V_k\}_{k=1}^r$ of $-L$, all $V_k$ are simultaneously triangularizable.
\item 
For some (equivalently: each) choice of Lindblad terms $\{V_k\}_{k=1}^r$ of $-L$, the $V_k$ generate a solvable Lie algebra.
\end{enumerate}
\end{theorem}
In particular if all Lindblad terms $V_k$ commute (e.g., if there is only one $V_k$), then the system is stabilizable for~\eqref{eq:simplex-control-system}.
If the system is stabilizable for~\eqref{eq:simplex-control-system}, then one can show that in the full system~\eqref{eq:bilinear-control-system}, every $\SU(n)$-orbit is approximately viable, i.e.~one can stay arbitrarily close to any orbit, see~\cite[Prop.~4.7]{MDES23}.

We will prove the theorem as a sequence of lemmas:

\begin{lemma}
\label{lemma:matrix-element-bound}
Let $A\in\R^{n\times n}$ be a matrix with non-positive values on the diagonal and non-negative values on the off-diagonal. Let $v\in\R^n$ be a vector with non-negative elements such that $Av=0$. For all $i\neq j$ we have that
$A_{ij}\,v_j\leq\|A\|_{\max}\,v_i$,
where $\|A\|_{\max}:=\max_{i,j}|A_{ij}|$.
\end{lemma}

\begin{proof}
Note that $(Av)_i=0$ is equivalent to $\sum_{k\neq i} A_{ik}v_k=-A_{ii}v_i=|A_{ii}|v_i$.
This yields
%
$A_{ij}v_j\leq  \sum_{k\neq i}A_{ik}v_k=|A_{ii}|v_i\leq\|A\|_{\max}v_i$,
for all $i\neq j$
as desired.
\end{proof}

\begin{lemma} 
If the system is stabilizable, then the Lindblad terms $\{V_k\}_{k=1}^r$ are simultaneously triangularizable.
\end{lemma}

\begin{proof}
For $\varepsilon>0$ small enough we consider the curve $\lambda:[0,\varepsilon]\to\Delta^{n-1}$ defined as
$\lambda(t) = (1- (t+t^2+t^3+\ldots+t^{n-1}), t, t^2, t^3, \ldots, t^{n-1})$.
By assumption, for all $\varepsilon>0$, there exists $-M_t\in\conv(\mf L)$ such that $-M_t\lambda(t)=0$.
Applying Lemma~\ref{lemma:matrix-element-bound} to $A=-M_t$
and $v=\lambda(t)$
we have that if $i>j$,
\begin{equation}\label{eq:ineq_Mij}
|(M_t)_{ij}|\leq m^\star\frac{\lambda_i(t)}{\lambda_j(t)}\to0
\end{equation}
as $t\to 0$ where $m^\star = \sup_t \|M_t\|_{\max}\leq\sup_{U\in\SU(n)}\|L_U\|_{\max}$ which is finite by compactness of $\SU(n)$ and continuity of $-L$.
Moreover, by compactness of $\conv(\mf L)$ we can pass to a subsequence of the $-M_t$ which converges to some $-M\in\conv(\mf L)$.
Eq.~\eqref{eq:ineq_Mij} then implies that $M$ is 
upper triangular. 
With this, Lemma~\ref{lemma:LU-zeros} shows that all Lindblad terms can be simultaneously triangularized.
\end{proof}

\begin{lemma} \label{lemma:triang-stab}
If all Lindblad terms are simultaneously triangularizable, then the system is stabilizable.
\end{lemma}

\begin{proof}
By assumption there exists $U\in\SU(n)$ such that $U^*V_kU$ is upper triangular for all $k=1,\ldots,r$.
Given any $m=1,\ldots,n$ we define $J^{(m)}(U)$ as the $m\times m$ block in the lower right corner of $J(U)$,
as well as $-L_U^{(m)}:=J^{(m)}(U)-\diag((J^{(m)}(U))^\top\e)$ (note that $L_U^{(m)}$ is a stochastic generator, but it is generally not a submatrix of $L_U$).
This definition yields an inductive structure:
because all $U^*V_kU$ are upper triangular, so is $J^{(m)}(U)$, meaning
for all $m=1,\ldots,n-1$ one has
\begin{equation*}
-L_U^{(m+1)}= \begin{pmatrix}
0&(J^{(m+1)}(U))_{1,2}&\cdots&\cdots&(J^{(m+1)}(U))_{1,m}\\
0&-(J^{(m+1)}(U))_{1,2}&0&\cdots&0\\
\vdots&\ddots&\ddots&\ddots&\vdots\\
\vdots&&\ddots&\ddots&0\\
0&\cdots&\cdots&0&-(J^{(m+1)}(U))_{1,m}
\end{pmatrix} +\begin{pmatrix}
0&0\\0&-L_U^{(m)}
\end{pmatrix}
\end{equation*}
as is readily verified.
Our goal is to show that for every $1\leq m\leq n$ and for every
$\lambda\in\Delta^{m-1}$ there exist $\ell\in\mathbb N$, $\mu\in\Delta^{\ell-1}$, and permutation matrices $P_1,\ldots,P_\ell\in\mathbb R^{m\times m}$ such that
$
-\sum_{i=1}^\ell\mu_iP_i^\top L_U^{(m)} P_i\lambda=0
$.
This would conclude the proof as then
$
-\sum_{i=1}^\ell\mu_iP_i^\top L_U^{(n)} P_i=-\sum_{i=1}^\ell\mu_i L_{U P_i}\in\conv\mf L
$
maps $\lambda$ to zero.

We proceed by induction on $m$. 
The case $m=1$ is trivial as $-L_U^{(1)}=0$.
For the induction step 
let any $1\leq m<n$ and any $\lambda\in\Delta^{m}$ be given.
We distinguish two cases:
if there exists a permutation matrix $P\in\mathbb R^{(m+1)\times (m+1)}$ such that $P\lambda=e_1$, then $-P^\top L_U^{(m+1)}P\lambda=0$ by the triangular structure of $L_U$.
If this is not the case, then the vector $\lambda^{(q)}:=(\lambda_1,\ldots,\lambda_{q-1},\lambda_{q+1},\ldots,\lambda_{m+1})/(1-\lambda_q)\in\Delta^{m-1}$ is well-defined for all $q=1,\ldots,m+1$.
By the induction hypothesis, for all $q=1,\ldots,m+1$ there exist
$\ell_q\in\mathbb N$, $\mu_q\in\Delta^{m-1}$, as well as permutation matrices 
$P_{q,1},\ldots, P_{q,\ell_q}\in\mathbb R^{m\times m}$ such that
$
-\sum_{i=1}^{\ell_q}(\mu_q)_iP_{q,i}^\top L_U^{(m)} P_{q,i}\lambda^{(q)}=0
$.
Defining $P'_{q,i}:=1\oplus P_{q,i}\in\mathbb R^{(m+1)\times(m+1)}$ for all $i=1,\ldots,\ell_q$, and $q=1,\ldots,m+1$,
the block structure of $-L_U^{(m+1)}$ from above
readily implies that
$
-\sum_{i=1}^{\ell_q}(\mu_q)_i(P_{q,i}')^\top L_U^{(m+1)} P_{q,i}'(\lambda_q,\lambda^{(q)}(1-\lambda_q))^\top$
has a non-negative first element and all other elements are non-positive.
Therefore
\begin{align*}
x_q:=&-\sum_{i=1}^{\ell_q}(\mu_q)_i (P_{k,i}'\pi_q)^\top L_U^{(m+1)} P_{q,i}' \pi_q\lambda\\
&=-\pi_q\Big(\sum_{i=1}^{\ell_q} (\mu_q)_i  (P_{q,i}')^\top L_U^{(m+1)} P_{q,i}' {\begin{pmatrix}
\lambda_q\\\lambda^{(q)}(1-\lambda_q)
\end{pmatrix}}\Big)
\end{align*}
for all $q=1,\ldots,m+1$ has a non-negative element in the $q^\text{th}$ component while all others are non-positive; here $\pi_q\in\mathbb R^{(m+1)\times (m+1)}$ which only swaps the first and the $q$-th component.
All that remains to do is to find $\xi\in\Delta^{m}$ such that $\sum_{q=1}^{m+1}\xi_q x_q=0$ as then
$
0=-\sum_{q=1}^{m+1}\sum_{i=1}^{\ell_q}\xi_q\mu_i \pi_q^\top (P_{q,i}')^\top L_U^{(m+1)} P_{q,i}' \pi_q \lambda
$
so $(\xi_k\mu_i)_{q,i}$ and $\{ P_{q,i}' \pi_q \,:\,q,i\}$ satisfy the property we are aiming to verify.

For the final step the key is that
$
X:=\begin{pmatrix}
x_1&\cdots&x_{m+1}
\end{pmatrix}\in\mathbb R^{(m+1)\times (m+1)}
$
satisfies $\e^\top X=0$ and $X_{ij}\geq 0$ for all $i\neq j$ by construction of the $x_q$. Therefore there exists $\epsilon>0$ such that $\mathds 1+\epsilon X$ is a stochastic matrix (if $X=0$ this is trivial, else choose $\epsilon:=(\max_j|X_{jj}|)^{-1}$).
In particular $(\mathds 1+\epsilon X)\Delta^m\subseteq\Delta^m$ so 
by the Brouwer fixed-point theorem~\cite{Brouwer11} this matrix has a fixed point $\xi\in\Delta^m$. But
this means $\xi+\epsilon X\xi=\xi$ which, due to $\epsilon>0$, is equivalent to $X\xi=0$, as desired.
\end{proof}

Note that if the $V_k$ are upper triangular in the standard basis, only permutations are used in the proof of Lemma~\ref{lemma:triang-stab} to stabilize a given state. Hence the proof also works for the toy model discussed in~\cite{CDC19,OSID23}.

\begin{lemma}
A set of matrices in $\C^{n\times n}$ is simultaneously triangularizable if and only if they generate a solvable Lie algebra $\mf g$.
\end{lemma}

\begin{proof}
By Lie's Theorem~\cite[Thm.~1.25]{Knapp02}, if $\mf g$ is solvable, then it is triangularizable.
Conversely, it is clear that a Lie algebra generated by simultaneously triangularizable matrices is solvable.
\end{proof}
This completes the proof of Theorem~\ref{thm:stab-system}. \smallskip

The theorem shows that determining whether a collection of square matrices
is simultaneously (unitarily)\footnote{
A family of complex matrices is triangularizable if and only if it is unitarily triangularizable. This can be seen, for instance, using the QR decomposition.} triangularizable
is a relevant task for studying our (reduced) control system.
Simultaneous triangularization certainly holds if
all matrices commute, but we can say more.
Many equivalent conditions are listed in~\cite[Ch.~1]{Radjavi00}.
Also recall that for a Lie algebra $\mf g\subseteq\mf{gl}(n,\C)$ it holds that $\mf g$ is strictly triangularizable if and only if every element of $\mf g$ is nilpotent (Engel's Theorem~\cite[Thm.~1.35]{Knapp02}). 

An interesting follow-up question would be to understand when every state in $\lambda\in\Delta^{n-1}$ is strongly stabilizable and how to find $U$ such that $-L_U\lambda=0$.

\subsection{Accessibility} \label{sec:accessible} 

As we will see below, the systems we consider are never (exactly) controllable. 
This is due to the dissipative nature of the dynamics. 
In this case, relevant notions are approximate controllability or accessibility.
Indeed, we will show that our systems are directly accessible almost everywhere.
The reduced system~\eqref{eq:simplex-control-system} is \emph{directly accessible} at some $\lambda\in\Delta^{n-1}$ if $\linspan(\derv(\lambda))=\R^n_0$. 
Intuitively this means that one can move directly in all directions of some cone with non-empty interior.
Using~\cite[Prop.~4.14 \& 4.15]{MDES23} we immediately obtain:

\begin{proposition}\label{prop:accessibility}
If $-L$ is non-unital, then the reduced control system~\eqref{eq:simplex-control-system} is generically\,\footnote{\label{footnote:generic}We say that a property holds generically on a set if it holds on an open, dense subset which has full measure.} directly accessible,
and hence the full bilinear system~\eqref{eq:bilinear-control-system} is generically accessible.
\end{proposition}

Direct accessibility in the unital case will be considered in Section~\ref{sec:unital-access}. 
The accessibility of the full system~\eqref{eq:bilinear-control-system} and its operator lift have been addressed in greater generality in~\cite{KDH12} using Lie-theoretic methods. 

\section{Reachability, Coolability and Controllability} \label{sec:reach}

One of the main questions in control theory is that of reachability, i.e.~given an initial state, what is the set of all states that can be generated within the control system?
Reaching pure states is of particular importance since it corresponds to coolability of the system.\footnote{Due to the assumption of fast unitary control, reaching a pure state, i.e.~a vertex of $\Delta^{n-1}$, is sufficient for reaching an energy minimizing state.}
Furthermore we will characterize which faces of the simplex $\Delta^{n-1}$ can be reached and when the system is controllable.

First we define some important notions. The \emph{reachable set} of $\lambda_0$ at time $T\geq0$ of the reduced control system~\eqref{eq:simplex-control-system}, denoted $\reach_{\ref{eq:simplex-control-system}}(\lambda_0,T)$, is the set of all $\lambda(T)$ where $\lambda:[0,T]\to\Delta^{n-1}$ is a solution to~\eqref{eq:simplex-control-system} with $\lambda(0)=\lambda_0$. 
The all-time reachable set is $\reach_{\ref{eq:simplex-control-system}}(\lambda_0)=\bigcup_{T\geq0}\reach_{\ref{eq:simplex-control-system}}(\lambda_0,T)$.
Moreover we say that $\lambda_f$ is \emph{approximately reachable} from $\lambda_0$ if $\lambda_f\in\overline{\reach_{\ref{eq:simplex-control-system}}(\lambda_0)}$ and it is \emph{asymptotically reachable} if there is a solution $\lambda:[0,\infty)\to\Delta^{n-1}$ with $\lambda(0)=\lambda_0$ and $\lambda(t)\to\lambda_f$ as $t\to\infty$.
The definitions for other control systems are entirely analogous.

The Equivalence Theorem~\ref{thm:equivalence} implies that the reachable sets of~\eqref{eq:simplex-control-system} and~\eqref{eq:bilinear-control-system} are closely related, see~\cite[Prop.~4.3]{MDES23} for a proof.

\begin{proposition} \label{prop:reach-equiv}
Given any $-L\in\wkl(n)$, 
let $\rho_0\in\mf{pos}_1(n)$ and $\lambda_0\in\Delta^{n-1}$ be such that $\spec^\down(\rho_0)=\lambda^\down_0$.
Then it holds that
\begin{align*}
\reach_{\ref{eq:bilinear-control-system}}(\rho_0,T) 
\subseteq 
\{U\lambda U^*:\lambda\in\reach_{\ref{eq:simplex-control-system}}(\lambda_0,T),U\in\SU(n)\}
\subseteq
\overline{\reach_{\ref{eq:bilinear-control-system}}(\rho_0,T)}.
\end{align*}
\end{proposition}

Let us continue with some general facts about reachable states. The following result extends some previous results derived in the toy model, see~\cite[Lem.~4 \& Lem.~6]{OSID23}. We omit the proof, since it is essentially the same.

\begin{proposition} \label{prop:closure-of-reach-contractible}
Let $\lambda\in\Delta^{n-1}$ and assume that $-L\in\wkl(n)$ is not purely Hamiltonian.
Then $\overline{\reach_{\ref{eq:relaxed-control-system}}(\lambda)}$ contains $\e/n$ and is contractible.
\end{proposition}

If a state can be reached asymptotically in the relaxed control system~\eqref{eq:relaxed-control-system}, then it is stabilizable. The proof is similar to the one of~\cite[Lem.~4.4~(ii)]{MDES23} and hence omitted.

\begin{lemma} \label{lemma:conv-stab}
Let $\lambda:[0,\infty)\to\Delta^{n-1}$ be a solution to~\eqref{eq:relaxed-control-system} such that 
$\lambda(t)\to\mu$ as $t\to\infty$ for some $\mu\in\Delta^{n-1}$.
Then $\mu\in\stab_{\ref{eq:simplex-control-system}}$.
\end{lemma}


A similar result also holds for faces.

\begin{lemma} \label{lemma:reach-viable}
Let $F$ be a face of $\Delta^{n-1}$ which is not a vertex.
Assume that there is some $\mu\in\relint(F)$ and some $\lambda\neq\mu$ such that $\mu\in\overline{\reach_{\ref{eq:relaxed-control-system}}(\lambda)}$.
Then $F$ is viable for~\eqref{eq:relaxed-control-system}.
\end{lemma}

\begin{proof}
We show the contrapositive, so assume that $F$ is not viable for~\eqref{eq:relaxed-control-system}.
By Proposition~\ref{prop:viable-faces} it holds for every $\mu\in\relint(F)$ that no element of $\derv(\mu)$ is tangent to $F$.
In particular, if $\alpha$ is any linear functional on $T_\mu\R^n_0$ 
which vanishes on tangent vectors along $F$ and is non-positive on the tangent cone $T_\mu\Delta^{n-1}$ (which is a subset of $T_\mu\R^n_0$), then $\alpha(\derv(\mu))\subset[-\infty,-\varepsilon]$ for some $\varepsilon>0$.
By continuity, and possibly shrinking $\varepsilon$, we may assume that $\alpha(\derv(\nu))\subset[-\infty,-\varepsilon]$ for every $\nu$ in some ball $B(R,\mu)$ with $R>0$.
Moreover, by compactness there is some $C>0$ such that $\derv(\nu)\subseteq B(C,0)$ for every $\nu$ in $B(R,\mu)$.
Hence there exists some $r>0$ small enough such that any solution to~\eqref{eq:relaxed-control-system} starting outside of $B(R,\mu)$ cannot enter $B(r,\mu)$.
In particular, for any $\lambda$ we can choose $R<\|\lambda-\mu\|$ and hence $\mu$ is not approximately reachable from $\lambda$.
\end{proof}

Note, however, that there are reachable states that are not stabilizable. This property and more can be shown explicitly in the toy model, see~\cite[Sec.~4.2 \& Fig.~7]{OSID23}. In general it is very difficult to compute the reachable set, hence why we will focus on the reachability of faces of the simplex (including coolability), and we will characterize approximate controllability, which is the situation where all states are approximately reachable from all other states.

\subsection{Asymptotically Coolable Systems}

One of the DiVincenzo criteria~\cite{VincCriteria} necessary for quantum computation requires the ability to initialize the system in a simple reference state.
In practice one often starts in a thermal state and cools the system to a temperature near absolute zero.
Since we assume fast unitary controllability, we will use the term ``cooling'' to refer to the preparation of any pure state.
In this section we characterize (asymptotic) coolability and show how it can be implemented in the bilinear control system~\eqref{eq:bilinear-control-system}.

We start with some technical results. If we think of the dynamics on the simplex $\Delta^{n-1}$ as a (controlled) Markov chain on $n$ states, the task of cooling corresponds to moving the entire population into a single state. To achieve this one has to minimize the outflow of this state and maximize the inflow. The next result yields some bounds on the ratio of inflow to outflow for certain systems.

\begin{lemma}
\label{lemma:bounded-flow-ratio}
Let $\{V_k\}_{k=1}^r$ denote a finite set of complex $n\times n$ matrices such that every common eigenvector of all $V_k$ is also a common eigenvector of all $V_k^*$. With $J(U)$ as in~\eqref{eq:def_JU}
there is a constant $C>0$ depending on the $V_k$ such that the flow ratio satisfies
\begin{align}\label{eq:flow-ratio}
\sum_{i=2}^n J_{1i}(U) \,\big/\, \sum_{i=2}^n J_{i1}(U) \leq C
\end{align}
for all $U$ where the expression is defined.
\end{lemma}

\begin{proof}
Our key object will be the subspace $S\subseteq\mathbb C^n$ defined as follows.
If there is no eigenvector which all the $V_k$ have in common set $S:=\{0\}$.
Else let $v_1$ denote such a common (unit) eigenvector of $V_k$ and set $S_1:=\linspan(v_1)$.
From this we proceed inductively over $i=1,\ldots,n$: 
if $S_i^\perp$ contains a common (unit) eigenvector $v_{i+1}$ we set $S_{i+1}=S_i\oplus\linspan(v_{i+1})$.
Else $S:=S_i$.
The advantage of this explicit construction is that---if $S\neq\{0\}$---we now have an orthonormal basis $\{v_i\}_{i=1}^s$ ($s\geq 1$) of $S$
consisting of common eigenvectors of all $V_k$.

Assume w.l.o.g.~that $S\neq\mathbb C^n$
(else the $V_k$ are all normal and commute with each other so the flow ratio equals $1$ where it is defined).
This assumption is equivalent to $S^\perp\neq\{0\}$ which means the domain of the following map is well defined:
\begin{align*}
b:\SU(n)\times\{\psi\in S^\perp\,.\,\|\psi\|=1\}&\to[0,\infty)\\
(U,\psi)&\mapsto \sum_{k=1}^r\big\| (V_k-\langle e_1|U^*V_kU|e_1\rangle\mathds1)\psi \big\|^2
\end{align*}
As $b$ is a continuous function on a compact domain it attains its minimum, that is, $C_D:=\min_{U\in\SU(n),\psi\in S^\perp,\|\psi\|=1}b(U,\psi)\geq 0$ exists.
Most importantly, $C_D>0$:
if $C_D$ were zero, then there exist $U\in\SU(n)$ and a normalized vector $\psi\in S^\perp$ such that $V_k\psi=\langle e_1|U^*V_kU|e_1\rangle\psi$ for all $k$.
Thus $\psi$ is a common eigenvector for all $V_k$ which contradicts $\psi\in S^\perp$.

With this let $U\in\SU(n)$ be given such that the flow ratio is well-defined (that is, $ \sum_{i=2}^n J_{i1}(U)>0$).
Our goal is to show that\footnote{
Here $\|\cdot\|_\infty$ denotes the usual operator norm, i.e.~the largest singular value of the input.
}
\begin{equation}\label{eq:lemma:bounded-flow-ratio_goal}
R(U)\leq\max\Big\{1,\frac{4\sum_{k=1}^r\|V_k\|_\infty^2}{C_D}\Big\}
\end{equation}
which would conclude the proof because the r.h.s.~of Eq.~\eqref{eq:lemma:bounded-flow-ratio_goal} is finite as we just saw.
Defining $w:=Ue_1$ and using Lemma~\ref{lemma:MU-row-column-sums}, a straightforward computation shows
\begin{align*}
R(U)
&=
\frac{\sum_{i=1}^n J_{1i}(U)-J_{11}(U)}{\sum_{i=1}^n J_{i1}(U)-J_{11}(U)}\\
&=
\frac{\sum_{k=1}^r \langle w|V_kV_k^*|w\rangle -|\langle w|V_k|w\rangle|^2}{\sum_{k=1}^r \langle w|V_k^*V_k|w\rangle -|\langle  w|V_k|w\rangle|^2}=
\frac{\sum_{k=1}^r\|(V_k^*-\langle w|V_k^*|w\rangle\mathds1)w\|^2}{\sum_{k=1}^r\|(V_k-\langle w|V_k|w\rangle\mathds1)w\|^2}\,.
\end{align*}
Next we use the orthogonal projection onto $S$---denoted $\Pi_S$---to split up the above norms 
via $\|\psi\|^2=\|\Pi_S\psi\|^2+\|(\mathds1-\Pi_S)\psi\|^2=\|\Pi_S\psi\|^2+\|\Pi_{S^\perp}\psi\|^2$.
Importantly, by assumption on $V_k$, both $V_k$ and $V_k^*$ commute with $\Pi_S$ (and thus with $\Pi_{S^\perp}$).
Thus we find that $R(U)$ is equal to
\begin{align*}
&\frac{\sum_{k=1}^r\|\Pi_S(V_k^*-\langle w|V_k^*|w\rangle\mathds1)w\|^2+\|\Pi_{S^\perp}(V_k^*-\langle w|V_k^*|w\rangle\mathds1)w\|^2}{\sum_{k=1}^r\|\Pi_S(V_k-\langle w|V_k|w\rangle\mathds1)w\|^2+\|\Pi_{S^\perp}(V_k-\langle w|V_k|w\rangle\mathds1)w\|^2}\\
&\ =\frac{\sum_{k=1}^r\|(V_k^*-\langle w|V_k^*|w\rangle\mathds1)\Pi_Sw\|^2+\|(V_k^*-\langle w|V_k^*|w\rangle\mathds1)\Pi_{S^\perp}w\|^2}{\sum_{k=1}^r\|(V_k-\langle w|V_k|w\rangle\mathds1)\Pi_Sw\|^2+\|(V_k-\langle w|V_k|w\rangle\mathds1)\Pi_{S^\perp}w\|^2}\\
&\ =\frac{\sum_{k=1}^r\|(V_k^*-\langle e_1|U^*V_k^*U|e_1\rangle\mathds1)\Pi_SUe_1\|^2+\|(V_k^*-\langle e_1|U^*V_k^*U|e_1\rangle\mathds1)\Pi_{S^\perp}Ue_1\|^2}{\sum_{k=1}^r\|(V_k-\langle e_1|U^*V_kU|e_1\rangle\mathds1)\Pi_SUe_1\|^2+\|(V_k-\langle e_1|U^*V_kU|e_1\rangle\mathds1)\Pi_{S^\perp}Ue_1\|^2}\,.
\end{align*}
At this point observe that $\|(V_k^*-\langle e_1|U^*V_k^*U|e_1\rangle\mathds1)\psi\|^2$ and $\|(V_k-\langle e_1|U^*V_kU|e_1\rangle\mathds1)\psi\|^2$ are equal for all $k$, all $U\in\SU(n)$, and all $\psi\in S$.
This can be seen, e.g., by expanding $\|\cdot\|^2$ into $\sum_{i=1}^s|\langle v_i|\,\cdot\,\rangle|^2$ (Parseval's identity)
and using that if $\lambda_{k,i}$ is the eigenvalue of $V_k$ w.r.t.~the eigenvector $v_i$, then\footnote{
By assumption $v_i$ is a normalized eigenvector of $V_k^*$ (to an eigenvalue $\mu_{k,i}$) which implies
$\overline{\lambda_{k,i}}=\langle V_k v_i|v_i\rangle=\langle v_i|V_k^*v_i\rangle=\mu_{k,i}$.
}
$V_k^*v_i=\overline{\lambda_{k,i}}v_i$.
Therefore the first summand of the numerator and the denominator of $R(U)$ coincide.
At this point we have to distinguish two cases: if $\Pi_{S^\perp}Ue_1=0$, then $R(U)=1$. Thus~\eqref{eq:lemma:bounded-flow-ratio_goal} holds and we are done.
Else we use the mediant inequality to obtain the (well-defined) expression
$$
R(U)\leq\max\Big\{1, \frac{\sum_{k=1}^r\|(V_k^*-\langle e_1|U^*V_k^*U|e_1\rangle\mathds1)\Pi_{S^\perp}Ue_1\|^2}{\sum_{k=1}^r\|(V_k-\langle e_1|U^*V_kU|e_1\rangle\mathds1)\Pi_{S^\perp}Ue_1\|^2}  \Big\}
$$
Finally we upper bound the second argument of this maximum as follows:
\begin{align*}
& \frac{\sum_{k=1}^r\|(V_k^*-\langle e_1|U^*V_k^*U|e_1\rangle\mathds1)\Pi_{S^\perp}Ue_1\|^2}{\sum_{k=1}^r\|(V_k-\langle e_1|U^*V_kU|e_1\rangle\mathds1)\Pi_{S^\perp}Ue_1\|^2} \\
 &\qquad=\frac{\sum_{k=1}^r\|(V_k^*-\langle e_1|U^*V_k^*U|e_1\rangle\mathds1)\frac{\Pi_{S^\perp}Ue_1}{\|\Pi_{S^\perp}Ue_1\|}\|^2}{\sum_{k=1}^r\|(V_k-\langle e_1|U^*V_kU|e_1\rangle\mathds1)\frac{\Pi_{S^\perp}Ue_1}{\|\Pi_{S^\perp}Ue_1\|}\|^2}  \\
 &\qquad\leq \frac{\sum_{k=1}^r\max_{\psi\in \mathbb C^n,\|\psi\|=1}\|(V_k^*-\langle e_1|U^*V_k^*U|e_1\rangle\mathds1)\psi\|^2}{\min_{\psi\in S^\perp,\|\psi\|=1}\sum_{k=1}^r\|(V_k-\langle e_1|U^*V_kU|e_1\rangle\mathds1)\psi\|^2} \\
 &\qquad\leq  \frac{\sum_{k=1}^r\|(V_k^*-\langle e_1|U^*V_k^*U|e_1\rangle\mathds1)\|_\infty^2}{\min_{\psi\in S^\perp,\|\psi\|=1}b(U,\psi)}\leq  \frac{\sum_{k=1}^r(2\|V_k\|_\infty)^2}{C_D} \,.\notag
\end{align*}
In either case~\eqref{eq:lemma:bounded-flow-ratio_goal} holds so we are done.
\end{proof}

The lemma above shows that under certain circumstances, the ``flow ratio'' of a set of matrices, as defined in~\eqref{eq:flow-ratio}, remains bounded, even if it is undefined at some points. This limits the ability of solutions to~\eqref{eq:relaxed-control-system} to converge to a vertex of $\Delta^{n-1}$:

\begin{corollary}
\label{coro:coolable-implies-eigenvector}
Let $\lambda:[0,\infty)\to\Delta^{n-1}$ be a solution to~\eqref{eq:relaxed-control-system}.
If every common eigenvector of the Lindblad terms $\{V_k\}_k$ is also a common eigenvector of $\{V_k^*\}_k$, then the first component $\lambda_1(t)\leq \max\{\lambda_1(0),\frac{C}{1+C}\}$ for all $t\geq 0$, where $C$ is as in~\eqref{eq:flow-ratio}.
Therefore, if $\lambda(0)\neq e_1$ and if there is a sequence $t_n\to\infty$ with $\lambda(t_n)$ converging to $e_1$, then there exists a common eigenvector of all Lindblad terms which is not a common left eigenvector.
\end{corollary}

\begin{proof}
It suffices to prove the first statement as the second one is an immediate consequence.
%
For this assume that every common eigenvector is also a left eigenvector. 
Indeed, a direct consequence of Lemma~\ref{lemma:bounded-flow-ratio} and the mediant inequality is that
$
\sum_{i=2}^n (-M_{1i})  \leq C\sum_{i=2}^n (-M_{i1})
$
for all $-M\in\operatorname{conv}\mf L$.
Note that this also holds when $\sum_{i=2}^n (-M_{i1})=0$.
Now let $\tilde\lambda\in\Delta^{n-1}$ and $-M\in\operatorname{conv}\mf L$ be given
such that
$(-M\tilde\lambda)_1>0$.
Then, because $(-M_{i1})\geq 0$ for all $i\neq 1$ we compute
\begin{align*}
0
&<
(-M_{11})\tilde\lambda_1 
+ \sum_{i=2}^n (-M_{1i} )\tilde\lambda_i 
=-\tilde\lambda_1 \sum_{i=2}^n (-M_{i1})
+ \sum_{i=2}^n (-M_{1i} )\tilde\lambda_i 
\\
&\leq-\tilde\lambda_1 \sum_{i=2}^n (-M_{i1})
+ \sum_{i=2}^n (-M_{1i} )\Big(\sum_{j=2}^n\tilde\lambda_j\Big)
= -\tilde\lambda_1 \sum_{i=2}^n(- M_{i1})
+ (1-\tilde\lambda_1) \sum_{i=2}^n(- M_{1i})\\
&\leq -\tilde\lambda_1 \sum_{i=2}^n(- M_{i1})
+ (1-\tilde\lambda_1)  C\sum_{i=2}^n (-M_{i1})=\Big(\sum_{i=2}^n (-M_{i1})\Big)(C-\tilde\lambda_1(1+C)) \,.
\end{align*}
Rearranging this inequality yields $\tilde\lambda_1<\frac{C}{1+C}$.
Now let $\lambda$ be any solution to~\eqref{eq:relaxed-control-system} and denote $C':=\max\{\lambda_1(0),\frac{C}{1+C}\}$.
Towards a contradiction assume that there is some $t_1>0$ such that $\lambda_1(t_1)>C'$.
Let $t_0=\max(\{t\in[0,t_1]: \lambda_1(t)=C'\})$.
Then $t_0<t_1$ and $\lambda_1(t)\geq C'$ for all $t\in[t_0,t_1]$, and by the above $(-M\lambda(t))_1\leq0$ for every $-M\in\conv(\mf L)$ on the same interval.
By Remark~\ref{rmk:reduced-aliter} there is $(M_t)_{t\in[t_0,t_1]}$ corresponding to $\lambda$. Then it holds that $\lambda_1(t_1)=\lambda_1(t_0)+\int_{t_0}^{t_1}(-M_t\lambda(t))_1\,dt\leq \lambda_1(t_0)=C'$, which yields the desired contradiction.
\end{proof}
This shows that determining upper bounds for the flow ratio~\eqref{eq:flow-ratio} of a system allows to find bounds on the purest reachable state.

Now we are ready to characterize asymptotically coolable systems.

\begin{theorem} 
\label{thm:asymptotic-coolability}
Given any $-L\in\wkl(n)$,
the following are equivalent.\smallskip
\begin{enumerate}[(i)]
\item \label{it:cool-evec} For each choice of Lindblad terms $\{V_k\}_{k=1}^r$ of $-L$,
there exists a common eigenvector of all $V_k$ which is not a common left eigenvector.
\item \label{it:cool-ham} There exists a (time-independent) Hamiltonian $H$ such that $-(\iu\ad_H + L)$ has a (unique) attractive fixed point\footnote{We say that $\rho$ is an attractive fixed point if every solution converges to $\rho$. If such an attractive fixed point exists, it is clearly unique.}, and this fixed point is pure.
\item \label{it:cool-conv} For every initial state, there exists some solution $\lambda$ converging to $e_1$.
\item \label{it:cool-reach-some} $e_1\in\overline{\reach(\lambda)}$ for some $\lambda\in\Delta^{n-1}\setminus\{e_1\}$.
\end{enumerate}
\end{theorem}

\begin{proof}
\ref{it:cool-evec} $\Rightarrow$ \ref{it:cool-ham} follows from Lemma~\ref{lemma:unique-attractive-fp},
\ref{it:cool-ham} $\Rightarrow$ \ref{it:cool-conv} follows from Corollary~\ref{coro:ham-control-system}, and
\ref{it:cool-conv} $\Rightarrow$ \ref{it:cool-reach-some} is trivial.
Finally \ref{it:cool-reach-some} $\Rightarrow$ \ref{it:cool-evec} is Corollary~\ref{coro:coolable-implies-eigenvector}.
\end{proof}

Most of our effort went into proving that if the system is coolable, then the Lindblad terms must have a common eigenvector which is not a common left eigenvector. The other implications are mostly known and have been rediscovered several times, see for instance~\cite{BNT08b,Kraus08,TV09}.
Theorem~\ref{thm:asymptotic-coolability} shows that for the purpose of cooling, time-independent controls are sufficient.
In the following section we will consider the reachability of faces, where time-dependent controls will offer new possibilities.

\subsection{Directly and Indirectly Reachable Faces}

In the previous section we characterized coolability by studying reachability of vertices of $\Delta^{n-1}$ in the reduced system. 
One consequence of Theorem~\ref{thm:asymptotic-coolability} is that for vertices, asymptotic and approximate reachability coincide. 
For interior points of higher-dimensional faces of $\Delta^{n-1}$ this need not hold anymore.
The reachability of faces of the simplex corresponds to the reachability of certain subspaces of Hilbert space.
Of particular interest might be the reachability of decoherence free subspaces~\cite{Lidar98DFS}.

Faces of $\Delta^{n-1}$ whose interior is reachable in an approximate sense satisfy the following dichotomy: Either the interior can be reached directly, or one first has to approach the boundary and then move parallel to the face, which we call indirect reachability. The precise result is given in Proposition~\ref{prop:dichotomy} at the end of the section.
Note that indirectly reachable faces provide a concrete problem for which time-independent controls are insufficient.

We start by characterizing direct reachability of faces.
Like in the previous section, flow ratios play an important role here, so we begin by formalizing the concept.

\begin{definition} \label{def:flows}
Let $J$ be an $n\times n$ matrix with non-negative off-diagonal entries.
Given any $d\in\{1,\ldots,n-1\}$,
the \emph{$d$-dimensional inflow} $f_{\sf in}^d$, the \emph{$d$-dimensional outflow} $f_{\sf out}^d$, and the \emph{$d$-dimensional flow ratio $R_d$} of $J$ are defined by
$$
f_{\sf in}^d:= \sum_{i=1}^d \sum_{j=d+1}^n J_{ij}\,, \quad
f_{\sf out}^d := \sum_{i=d+1}^n \sum_{j=1}^d J_{ij}\,, \quad
R_d := \frac{f_{\sf in}^d}{f_{\sf out}^d}\,,
$$
respectively.
If the matrix $J$ is of the form $J(U)$ as in~\eqref{eq:def_JU}, then we denote the objects above by $f_{\sf in}^d(U)$, $f_{\sf out}^d(U)$ and $R_d(U)$\footnote{
Since $-L_U$ is well-defined, so are the off diagonal elements of $J(U)$, and hence also all the quantities defined here.
}.
We allow $R_d(U)$ to take values in $[0,+\infty]$, where expressions of the form $\frac{c}{0}$ with $c>0$ are interpreted as $+\infty$. Only expressions of the form $\frac00$ are considered undefined.\marginpar{could we not replace ``undefined'' with ``zero''? because then we don't have to assume anything implicit about sup and inf(?)}\marginpar{if anything it should be 1, but I would prefer not to (will think about it)}
We will say that the system has \emph{bounded $d$-dimensional flow ratio} if 
$\sup\{R_d(U)\,:\,U\in\SU(n)\} <\infty\,,$
otherwise we say it is \emph{unbounded}. Here and henceforth, suprema and infima of this form always implicitly ignore undefined values.
\end{definition}

Note that the flow ratio $R_d(U)$ is infinite or undefined whenever the first $d$ columns of $U^*$ span a lazy subspace, cf.~Lemma~\ref{lemma:lazy}~\ref{it:unitary}. 
Outside of these points, however, $R_d$ is a continuous function. 
The behavior of $R_d(U)$ near these singularities has important consequences for the reachability of faces, see Remark~\ref{rmk:directly-reach}.

In the proof of Corollary~\ref{coro:coolable-implies-eigenvector} we showed that if the $1$-dimensional flow ratio of the system is bounded, then it is impossible to reach a vertex of $\Delta^{n-1}$. In higher dimensions the situation becomes more nuanced, as a bounded flow ratio only prohibits approaching the interior of a face directly, but it does not prohibit approaching the boundary of the face:

\begin{lemma} \label{lemma:bounded-s-flow-ratio}
Let $1\leq d\leq n-1$ and set $p_d:\Delta^{n-1}\to\R$, $p_d(\lambda):=\sum_{i=1}^d\lambda_i$.
Given $\lambda\in\Delta^{n-1}$, $\epsilon>0$ assume that $\lambda_i\geq\varepsilon p_d(\lambda)$ for all $i=1,\ldots,d$,
and that the system has bounded flow ratio $R_d(U)$ with $R$ denoting the supremum.
Then if $p_d(\lambda)\geq\frac{R}{R+\varepsilon}$,
it holds that $p_d(\derv(\lambda))\subseteq(-\infty,0]$.
In particular, no solution to~\eqref{eq:relaxed-control-system} can converge to the interior of a $(d-1)$-dimensional face if it starts outside of the face.
\end{lemma}

\begin{proof}
Using~\eqref{eq:J-comp} we
for any $U\in\SU(n)$ compute
\begin{align*}
p_d(-L_U\lambda)
&=\sum_{i=1}^d\sum_{j=1}^nJ_{ij}(U) \lambda_j-\sum_{i=1}^n\sum_{j=1}^dJ_{i j}(U)\lambda_j\\
&=\sum_{i=1}^d\sum_{j=d+1}^nJ_{ij}(U) \lambda_j-\sum_{i=d+1}^n\sum_{j=1}^dJ_{i j}(U)\lambda_j\\
&\leq  (1-p_d(\lambda))f_{\sf in}^d(U)-\epsilon p_d(\lambda)f_{\sf out}^d(U)\,,
\end{align*}
where we used $J_{ij}(U)\geq 0$ for all $i\neq j$, as well as $\lambda_i\geq\varepsilon p_d(\lambda)$ for all $i=1,\ldots,d$.
We distinguish two cases: 
If $f_{\sf out}^d(U)=0$, then $f_{\sf in}^d(U)=0$. The reason for this is that $R$ is a continuous, and by assumption bounded, function whose zero set is either all of $\U(n)$ or nowhere dense (as in the proof of Lemma~\ref{lemma:generate-matrix-algebra})
so if $f_{\sf in}^d(U)$ were not zero for some $U$ (while $f_{\sf out}^d(U)$ is), then $R$ cannot be bounded.
Thus $p_d(\derv(\lambda))\subseteq(-\infty,0]$ either way.
Now if $f_{\sf out}^d(U)\neq 0$, then the 
above estimate is non-positive if and only if $p_d(\lambda)\geq1-\frac{\epsilon}{R_d(U)+\varepsilon}$. 
Hence if $p_d(\lambda)\geq\frac{R}{R+\varepsilon}=1-\frac{\epsilon}{R+\varepsilon}$, then $-p_d(L_U\lambda)\leq0$ for all $U\in\SU(n)$.
Now consider any solution $\lambda:[0,\infty)\to\Delta^{n-1}$ to~\eqref{eq:relaxed-control-system} and let $F$ be the 
convex hull of the $d$ first standard basis vectors.
If $\lambda$ converges to some $\mu\in\relint(F)$, then for $t$ large enough and some $\varepsilon>0$ it holds that $\lambda_i(t)\geq\varepsilon p_d(\lambda(t))$ for all $i=1,\ldots,d$. At the same time $p_d(\lambda(t))$ converges to $1$.
However, this contradicts the first part of this lemma as in the proof of Corollary~\ref{coro:coolable-implies-eigenvector}.
\end{proof}

\begin{remark}
We have shown that if $p_d(\lambda)\geq\frac{R}{R+\varepsilon}$ then $p_d(\dot\lambda)\leq0$. 
This however does not necessarily imply that for $p_d(\lambda)>\frac{R}{R+\varepsilon}$ we have $p_d(\dot\lambda)<0$. 
For instance one might consider a system with a single normal $V$ which is not a multiple of the identity (see Coro.~\ref{lemma:unital-stabilizable}). 
In this case every state is stabilizable but the one-dimensional flow ratio is equal to $1$ whenever it is defined. 
\end{remark}

With this result at our disposal, we can start characterizing the direct reachability of faces in the simplex.
We begin by considering a stronger notion, i.e.~which faces of $\Delta^{n-1}$ can be reached using a time-independent Hamiltonian, as shown in~\cite{TV09}.

\begin{lemma} \label{lemma:ham-reachable}
The following statements are equivalent. \smallskip
\begin{enumerate}[(i)]
\item \label{it:ham-reach-lazy-enc} There exists a $d$-dimensional lazy subspace which is not an enclosure (recall Lemma~\ref{lemma:enclosure}).
\item \label{it:ham-reach-fix-pt} There exists a state $\rho$ of rank $d$ and a Hamiltonian $H$ such that $\rho$ is the unique fixed point of $-(\iu\ad_H+L)\in\wkl(n)$, so in particular $\rho$ is attractive.
\end{enumerate}
\end{lemma}

\begin{proof}
\ref{it:ham-reach-lazy-enc} $\Rightarrow$ \ref{it:ham-reach-fix-pt}: 
Follows from Lemma~\ref{lemma:unique-attractive-fp}. 
\ref{it:ham-reach-fix-pt} $\Rightarrow$ \ref{it:ham-reach-lazy-enc}:
Let $S=\supp(\rho)$. 
Since $\rho$ is a fixed point, $S$ is collecting by Lemma~\ref{lemma:conf-fixed}, and hence lazy.
If $S$ was an enclosure, there would be at least two fixed points (one supported on $S$, and one on $S^\perp$).
Since $-(\iu\ad_H+L)$ has a unique fixed point $\rho$, by~\cite[Thm.~18]{BNT08b} the generator has only one eigenvalue equal to $0$ corresponding to the fixed point, and all other eigenvalues have strictly negative real part.
This shows that $\rho$ is attractive.
\end{proof}

Finally we can give the proper notion of directly reachable faces and characterize them via unbounded flow ratios:

\begin{proposition}
\label{prop:directly-reachable-faces}
Let $1\leq d\leq n-1$ and let $F$ be a $(d-1)$-dimensional face of $\Delta^{n-1}$.
Then the following are equivalent:\smallskip
\begin{enumerate}[(i)]
\item \label{it:dir-reach-flow} The $d$-dimensional flow ratios $R_d(U)$ are unbounded.
\item \label{it:dir-reach-conv} There exists a solution $\lambda:[0,\infty)\to\Delta^{n-1}$ to~\eqref{eq:relaxed-control-system} with initial state $\lambda(0)\not\in F$ such that $\lim_{t\to\infty}\lambda(t)=:\lambda_F\in\relint(F)$.\smallskip
\end{enumerate}
In this case we say that $F$ is \emph{directly reachable}.
\end{proposition}

\begin{proof} 
\ref{it:dir-reach-conv} $\Rightarrow$ \ref{it:dir-reach-flow}: Immediate from Lemma~\ref{lemma:bounded-s-flow-ratio}.
\ref{it:dir-reach-flow} $\Rightarrow$ \ref{it:dir-reach-conv}: 
Due to the permutation symmetry of~\eqref{eq:relaxed-control-system} we may assume that $F=\conv(e_1,\ldots,e_d)$.
By assumption there is a sequence $(U_i)_{i=1}^\infty$ such that $R_d(U_i)\to\infty$ as $i\to\infty$.
Let $S_n^F$ be the subgroup of permutation matrices which map $F$ to itself and let $A\subseteq\Delta^{n-1}$ be the line segment consisting of all points fixed by $S_n^F$.
We can parametrize $A$ via $\iota:[0,1]\to A$, $a\mapsto \iota(a):=(\tfrac{a}{d},\ldots,\tfrac{a}{d},\tfrac{1-a}{n-d},\ldots\tfrac{1-a}{n-d})$. 
Note that $R_d(U_i)=R_d(U_iP)$ for all $P\in S_n^F$.
If we set $M_i=\frac{1}{|S_n^F|}\sum_{P\in S_n^F}P^\top L_{U_i}P=\frac{1}{|S_n^F|}\sum_{P\in S_n^F}L_{U_iP}$, then all $-M_i\in\operatorname{conv}\mf L$ and they leave $A$ invariant.
Using
$p_d:\Delta^{n-1}\to\R:\lambda\mapsto\sum_{i=1}^d\lambda_i$ 
we compute the derivative along $A$ as 
$-p_d(M_i\iota(a))=-p_d(L_{U_i}\iota(a))=-\tfrac{a}{d}f^d_{\sf out}(U_i)+\tfrac{1-a}{n-d}f^d_{\sf in}(U_i)$.
Hence the unique (attractive) fixed point on $A$ is located at $a=\tfrac{dR_d(U_i)}{dR_d(U_i)+n-d}$ which can be made arbitrarily close to $1$ for $i$ large enough.
This shows that there exists a solution to~\eqref{eq:relaxed-control-system} converging to $\mu:=\iota(1)\in A\cap F$.
\end{proof}

As a consequence of Lemma~\ref{lemma:ham-reachable} together with Corollary~\ref{coro:ham-control-system} we obtain the following:

\begin{corollary} \label{coro:directly-reachable-faces}
Let $F$ be a face of $\Delta^{n-1}$ of dimension $d-1$. 
If there exists a lazy subspace of dimension $d$ which is not an enclosure, then $F$ is directly reachable.
\end{corollary}

\begin{remark} \label{rmk:directly-reach}
Currently, we do not know whether the converse to Corollary~\ref{coro:directly-reachable-faces} also holds.
A possible generalization of Lemma~\ref{lemma:bounded-flow-ratio} to invariant subspaces of dimension higher than $d=1$ could be used to prove this, but the proof of said lemma does not seem to generalize in a straightforward manner.
If the converse does not hold, this would imply that time-dependent Hamiltonians allow one to directly reach faces not directly reachable using time-independent Hamiltonians.
\end{remark}

The following corollaries yield special cases where the converse of Corollary~\ref{coro:directly-reachable-faces} does hold.

\begin{corollary}
Let $F$ be a face of $\Delta^{n-1}$ of dimension $d-1$. 
If the commutant of $\{H_0,V_1,\ldots,V_r\}$ contains only multiples of the identity and $F$ is directly reachable, then there exists a lazy subspace of dimension $d$ which is not an enclosure.
\end{corollary}

\begin{proof}
By Lemma~\ref{lemma:enclosure}~\ref{it:commute} there are no proper non-trivial enclosures, and since $F$ is directly reachable, by Lemma~\ref{lemma:conv-stab} and Proposition~\ref{prop:viable-faces} there is a lazy subspace of dimension $d$.
\end{proof}

For the second corollary we observe the following duality relations for flow ratios:

\begin{lemma} \label{lemma:flow-duality}
Let $-L\in\wkl(n)$ and
some choice of Lindblad terms $\{V_k\}_{k=1}^r$ of $-L$ be given,
and let $-L'\in\wkl(n)$ be the Kossakowski--Lindblad generator represented by the Lindblad terms $\{V_k^*\}_{k=1}^r$.
Then the following hold (here we use $1/0=+\infty$ and the suprema and infima are defined as in Definition~\ref{def:flows}):
\begin{enumerate}[(i)]
\item \label{it:dual-codim} $\sup_{U\in\SU(n)} R_{n-d}(U) = (\inf_{U\in\SU(n)} R_{d}(U))^{-1}$
\item \label{it:dual-trans} $\sup_{U\in\SU(n)} R'_d(U) = (\inf_{U\in\SU(n)} R_d(U))^{-1}$\smallskip
\end{enumerate}
where $R_d(U)$ is the flow ratio of $-L$ and $R_d'(U)$ is the flow ratio of $-L'$.
\end{lemma}

\begin{proof}
\ref{it:dual-codim}: Follows from the observation that $f_{\sf in}^d(U)=f_{\sf out}^{n-d}(\pi U\pi^\top)$ and vice-versa where $\pi=\sum_{i=1}^n e_{n+1-i}e_{i}^\top$.
\ref{it:dual-trans}: 
Replacing all $V_k$ by $V_k^*$ transforms the $d$-dimensional inflow into the
$d$-dimensional outflow.
\end{proof}

\begin{corollary}
If $F$ is a facet\footnote{A facet of $\Delta^{n-1}$ is a face of dimension $n-2$.} of $\Delta^{n-1}$ and if $F$ is directly reachable, then exists a lazy subspace of dimension $n-1$ which is not an enclosure.
\end{corollary}

\begin{proof}
Proposition~\ref{prop:directly-reachable-faces} shows that the $(n-1)$-dimensional flow ratio is unbounded.
Towards a contradiction assume that all $n-1$ dimensional lazy subspaces are enclosures.
Consider the system $-L'$ with Lindblad terms $\{V_k^*\}_{k=1}^r$.
Then all common eigenvectors of the $V_k^*$ are also common eigenvectors of all $V_k$.
By Lemma~\ref{lemma:bounded-flow-ratio} it holds that $\sup_{U\in\SU(n)} R_1'(U)<C$ for some $C<\infty$.
Then by Lemma~\ref{lemma:flow-duality}~\ref{it:dual-codim} and~\ref{it:dual-trans} (with $d=1$) we have $\sup_{U\in\SU(n)} R_{n-1}(U)=\sup_{U\in\SU(n)} R_1'(U)$,
i.e.~the $(n-1)$-dimensional flow ratio is bounded, yielding the desired contradiction.
\end{proof} 

\begin{corollary}
If $n=3$, then the converse of Corollary~\ref{coro:directly-reachable-faces} holds for all faces.
\end{corollary}

Finally we can prove the promised dichotomy alluded to in the beginning. We say that a face $F$ of $\Delta^{n-1}$ is approximately reachable if there exists $\lambda\notin F$ such that $\overline{\reach(\lambda)}\cap \relint(F)\neq\emptyset$. A face that is approximately reachable but not directly reachable is called \emph{indirectly reachable face}.

\begin{proposition} 
\label{prop:dichotomy}
Let $F$ be a face of $\Delta^{n-1}$. 
Consider the following statements:\smallskip
\begin{enumerate}[(i)] 
\item \label{it:approx-sufficient} $F$ is directly reachable, otherwise $F$ is viable but not purely Hamiltonian\footnote{\label{footnote:purely-ham-restr}
For a viable face $F$ there exists at least some $-L_U\in\mf L$ whose restriction to $F$ is tangent to $F$. If every such vector field vanishes on $F$, we say that $F$ is purely Hamiltonian.} and some lower dimensional face is approximately reachable.
\item \label{it:approx-reach} $F$ is approximately reachable.
\item \label{it:approx-dir-indir} $F$ is directly reachable, otherwise $F$ is viable and some lower dimensional face is approximately reachable.
\end{enumerate} \smallskip
Then we have the implications: \ref{it:approx-sufficient} $\Rightarrow$ \ref{it:approx-reach} $\Rightarrow$ \ref{it:approx-dir-indir}.
\end{proposition}

\begin{proof}
\ref{it:approx-sufficient} $\Rightarrow$ \ref{it:approx-reach}: 
If $F$ is directly reachable then~\ref{it:approx-reach} is clearly true. Otherwise,
some point on the boundary is asymptotically reachable, and since the face is viable and not purely Hamiltonian we can reach, for instance, the center of $F$ as in Proposition~\ref{prop:closure-of-reach-contractible}.
\ref{it:approx-reach} $\Rightarrow$ \ref{it:approx-dir-indir}:
If $F$ is not directly reachable, then the flow ratio is bounded by Proposition~\ref{prop:directly-reachable-faces}.
So in order to approach an interior point of $F$, one must increase the value of $p_s(\lambda)$ (as defined in Lemma~\ref{lemma:bounded-s-flow-ratio}), which implies by compactness that some point on the boundary of $F$ is approximately reachable, and hence some face of lower dimension is approximately reachable. 
If $F$ is not viable, then no interior point of $F$ is approximately reachable from any other point, cf.~Lemma~\ref{lemma:reach-viable}.
\end{proof}
The reason we do not obtain equivalence is that even if every lazy subspace of appropriate dimension has purely Hamiltonian dynamics, it might still be possible to move along the face with arbitrarily small outflow.
This can be made rigorous using flow ratio arguments as above, but we will not do so here.

\subsection{Reverse Coolable and Controllable Systems}

So far we have studied under which conditions faces of the simplex can be reached, with special emphasis on reachability of vertices.
In this section we try to understand under which conditions all states in $\Delta^{n-1}$ can be reached. More precisely, if for all $\lambda,\mu\in\Delta^{n-1}$ it holds that $\mu\in\overline{\reach_{\ref{eq:simplex-control-system}}(\lambda)}$, then we say that the system~\eqref{eq:simplex-control-system} is \emph{approximately controllable}. (Recall that~\eqref{eq:simplex-control-system} and ~\eqref{eq:relaxed-control-system} have the same reachable sets after taking the closure.) Note that since one can never exactly reach the boundary of the simplex from the interior in finite time, the system is never controllable in the usual sense~\cite[Thm.~3.10]{DiHeGAMM08}.

It turns out to be useful to consider time-reversed dynamics on the simplex.
Note that under such dynamics $\Delta^{n-1}$ ceases to be forward invariant.
We say that the reduced control system is \emph{reverse coolable} if $\overline{\reach_{\ref{eq:simplex-control-system}}(e_1)} \supseteq \Delta_\down^{n-1}$.
The reason for introducing this artificial concept is that asymptotic coolability together with reverse coolability characterizes controllable systems.
The results of this section generalize results on approximate controllability of a quantum system coupled to a heat bath of temperature zero, cf.~\cite[Thm.~1 \& 2]{CDC19} and~\cite{BSH16}.

\begin{proposition} \label{prop:approx-ctrl}
The following are equivalent: \smallskip
\begin{enumerate}[(i)]
\item \label{it:approx-contr} The system~\eqref{eq:simplex-control-system} is approximately controllable.
\item \label{it:bicoolable} The system~\eqref{eq:simplex-control-system} is asymptotically coolable and reverse coolable.
\item \label{it:approx-contr-full} The system~\eqref{eq:bilinear-control-system} is approximately controllable.
\end{enumerate}
\end{proposition}

\begin{proof}
Consider the reachability relation defined by $\lambda\rightsquigarrow\mu$ if $\mu\in\overline{\reach_{\ref{eq:simplex-control-system}}(\lambda)}$.
One can show that this is a preorder, in particular, that it is transitive. 
\ref{it:approx-contr} $\Rightarrow$ \ref{it:bicoolable}:
Clearly $\lambda\rightsquigarrow e_1$ for all $\lambda\in\Delta^{n-1}$ and hence by Theorem~\ref{thm:asymptotic-coolability} the system is asymptotically coolable. Reverse coolability is clear.
\ref{it:bicoolable} $\Rightarrow$ \ref{it:approx-contr}:
Again by Theorem~\ref{thm:asymptotic-coolability} $\lambda\rightsquigarrow e_i$ for all $\lambda\in\Delta^{n-1}$ and $i=1,\ldots,n$.
By reverse coolability $e_1\rightsquigarrow\lambda$ for all $\lambda\in\Delta_\down^{n-1}$. By permutation symmetry of the system every point in $\Delta^{n-1}$ is reachable from some vertex and by transitivity the system is approximately controllable.
\ref{it:approx-contr} $\Leftrightarrow$ \ref{it:approx-contr-full}: Direct consequence of~\cite[Prop.~4.16]{MDES23}.
\end{proof}

This result motivates us to characterize reverse coolability, generalizing a toy model result~\cite[Lem.~3]{CDC19}:

\begin{proposition} 
\label{prop:rev-coolable}
Consider the following statements: \smallskip
\begin{enumerate}[(i)]
\item \label{it:rev-cool-inv} All faces (except possibly vertices) of $\Delta^{n-1}$ are viable for~\eqref{eq:relaxed-control-system} but not purely Hamiltonian (cf.~footnote~\ref{footnote:purely-ham-restr}).
\item \label{it:rev-cool-reach} The system \eqref{eq:simplex-control-system} is reverse coolable. 
\item \label{it:rev-cool-lazy} All faces (except possibly vertices) of $\Delta^{n-1}$ are viable for~\eqref{eq:relaxed-control-system}. \smallskip
\end{enumerate}
Then we have the following implications:
\ref{it:rev-cool-inv} $\Rightarrow$ \ref{it:rev-cool-reach} $\Rightarrow$ \ref{it:rev-cool-lazy}.
\end{proposition}

\begin{proof}
\ref{it:rev-cool-inv} $\Rightarrow$ \ref{it:rev-cool-reach}:
Let $F$ be any face of dimension $d-1$ of the simplex $\Delta^{n-1}$ which is not a vertex ($d>1$), and let $\lambda_d\in F$ be arbitrary.
First we show that there is some $\lambda_{d-1}\in\partial F$ on the boundary such that $\lambda_{d-1}\rightsquigarrow\lambda_d$.
If $\lambda_d\in\partial F$ we set $\lambda_{d-1}=\lambda_d$.

By assumption there is some $L_U$ such that $F$ is invariant but not fixed. 
If $S_n^F$ denotes the permutation subgroup which leaves $F$ invariant, then $M=\frac{1}{|S_n^F|}\sum_{P\in S_n^F}{L_{UP}}$ still leaves $F$ invariant with the center being the unique attractive fixed point (cf.~\cite[Lem.~6]{OSID23}). 
Since the entire boundary of $F$ converges to its center, it passes through every point\footnote{This can be shown rigorously using the fact that $\pi_n(S^n)=\mathbb Z$, where $\pi_n$ denotes the $n$-th homotopy group~\cite[Sec.~4.1]{Hatcher02}.} of $F$, and hence every point is approximately reachable from the boundary. 
To be precise we have proven the claim in~\eqref{eq:relaxed-control-system}, but by the relaxation theorem (cf.~\cite[Ch.~2.4, Thm.~2]{Aubin84}),
it still holds in~\eqref{eq:simplex-control-system}.
Putting everything together, if we start with any $\lambda\in\Delta_\down^{n-1}$, we set $\lambda_n=\lambda$ and find a sequence of $\lambda_d$ for $d=n-1,\ldots,1$ where necessarily $\lambda_1=e_i$ for some $i=1,\ldots,n$.
By transitivity we get that $e_i\rightsquigarrow\ldots\rightsquigarrow\lambda$.
Using the permutation symmetry~\cite[Lem.~A.2]{MDES23} and forcing the solution to stay in the ordered Weyl chamber~\cite[Prop.~A.4]{MDES23} we find that $\lambda=\lambda^\down\in\overline{\reach_{\ref{eq:simplex-control-system}}(e_1)}$.

\ref{it:rev-cool-reach} $\Rightarrow$ \ref{it:rev-cool-lazy}:
Let $\lambda\in\relint(F)$ and assume that $F$ is not viable. 
Then by Lemma~\ref{lemma:reach-viable}, $\lambda$ is not approximately reachable from any other point.
\end{proof}

For stabilizable systems we can strengthen the result above:

\begin{corollary} 
If~\eqref{eq:simplex-control-system} is stabilizable and has a two-dimensional lazy subspace which is not purely Hamiltonian, then it is reverse coolable.
\end{corollary}

\begin{proof}
By Theorem~\ref{thm:stab-system} the Lindblad terms are simultaneously triangularizable.
From~\cite[Lem.~1.5.2]{Radjavi00} and the assumption it follows that we can choose a triangularization such that the two-dimensional subspace in the chain is not purely Hamiltonian.
Hence we satisfy Proposition~\ref{prop:rev-coolable}~\ref{it:rev-cool-inv}.
\end{proof}

\section{Special Structure for Unital Systems} \label{sec:unital}

In this section we focus on unital systems, which are precisely those systems for which the identity is a fixed point or, equivalently, on the level of generators, those which satisfy $L(\mathds1)=0$. 
Such systems stand in contrast to coolable and controllable systems studied above, since all reachable states are majorized by the initial state. 
In particular the purity of a state can never increase.

\begin{lemma} 
\label{lemma:unital-systems}
Given any $-L\in\wkl(n)$, the following statements are equivalent.\smallskip
\begin{enumerate}[(i)]
\item \label{it:unital} $-L$ is unital, i.e.~$L(\mathds1)=0$.
\item \label{it:unital-commutator} $\sum_{k=1}^r[V_k,V_k^*]=0$ for some (equivalently: each) choice of Lindblad terms $\{V_k\}_{k=1}^r$ of $-L$.
\item \label{it:unital-MU} $L_U\e = 0 $ for all $U\in\SU(n)$.
\item \label{it:unital-JU} $J_U\e = J_U^\top\e$ for all $U\in\SU(n)$.
\item \label{it:unital-lambda} $\derv(\e/n) = \{0\}$.
\item \label{it:unital-reach} $\overline{\reach_{\ref{eq:relaxed-control-system}}(\lambda)} \subseteq \{\mu : \mu \prec \lambda \}$ for all $\lambda \in \Delta^{n-1}$.
\end{enumerate}
\end{lemma}

\begin{proof}
A direct computation shows $-L(\mathds1)=\frac12\sum_{k=1}^r[V_k,V_k^*]$ for any choice of Lindblad terms $\{V_k\}_{k=1}^r$ of $-L$; hence~\ref{it:unital} $\Leftrightarrow$ \ref{it:unital-commutator}.
For \ref{it:unital-MU} $\Leftrightarrow$ \ref{it:unital-JU} $\Leftrightarrow$ \ref{it:unital-lambda}
note that by definition $-L_U=J(U)-\diag(J(U)^\top\e)$, which shows that $L_U\e=J(U)\e-J(U)^\top\e$. 
Next, Lemma~\ref{lemma:MU-row-column-sums} together with the Schur--Horn theorem \cite{Schur23,Horn54} shows that $\derv(\e/n)$ is equal to the majorization polytope spanned by the vector of eigenvalues of $\sum_{k=1}^r[V_k,V_k^*]$;
this shows the equivalence of~\ref{it:unital-commutator} and~\ref{it:unital-lambda}. 
Condition~\ref{it:unital-MU} shows that the time evolution of a unital system is doubly stochastic, which implies~\ref{it:unital-reach} see~\cite[Thm.~A.4]{MarshallOlkin}. Conversely, \ref{it:unital-reach} directly implies~\ref{it:unital-lambda}.
(The equivalence of~\ref{it:unital} and \ref{it:unital-reach} was also shown in~\cite{Yuan10}.)
\end{proof}
Importantly, unitality is independent of the Hamiltonian part of the generator, and hence of the control.
Note also that property~\ref{it:unital-reach} continues to hold
in the infinite-dimensional setting under appropriate assumptions~\cite{OSID19}.

Recall that $\mc V=\generate{\mathds1,V_1,\ldots,V_r}{alg}$---i.e.~the matrix algebra generated by the Lindblad terms and the identity---is called the relaxation algebra of $-L$,
and it is independent of the choice of Lindblad terms (Lemma~\ref{lemma:lindblad-alg-well-def}).
Moreover, we write $\Lat(\mc V)$ for the lattice of invariant subspaces for $\mc V$ (i.e.~its lazy subspaces), cf.~Appendix~\ref{app:lindblad-eq}.
For unital systems, stabilizability and reachability are to a large extent characterized by the structure of the relaxation algebra, which turns out to be particularly nice as the following result shows:

\begin{lemma}
If $-L$ is unital, then the corresponding lattice $\Lat(\cV)$ of lazy subspaces is orthocomplemented\footnote{This means that for every invariant subspace $S\in\Lat(\cV)$, the orthocomplement $S^\perp$ is also invariant.\label{footnote_orthocomplemented}
} and $\cV$ is a $*$-algebra\footnote{$\cV$ is a $*$-algebra if for every $V\in\cV$ it holds that $V^*\in\cV$.}.
\end{lemma}

\begin{proof}
Let $S\subseteq\mathbb C^n$ be any lazy subspace of $-L$; w.l.o.g.~$S\neq\{0\}$, $S\neq\mathbb C^n$.
Choose a unitary $U$ such that the first $k$ columns of $U$ span $S$. Then by Lemma~\ref{lemma:lazy}~\ref{it:unitary} the matrix $-L_U$ (and thus $J_U$) is block triangular.
Using Lemma~\ref{lemma:unital-systems}~\ref{it:unital-JU} we compute
\begin{align*}
0&=\sum_{j=1}^{\operatorname{dim}S}\big( (J_U\e )_j-(J_U^\top\e)_j \big)
=\sum_{j=1}^{\operatorname{dim}S}\sum_{k=\operatorname{dim}S+1}^n J_{jk}(U)-\sum_{j=1}^{\operatorname{dim}S}\sum_{k=\operatorname{dim}S+1}^n J_{kj}(U)\,;
\end{align*}
but the second term vanishes due to $J_U$ being block triangular.
Thus, because $J_{jk}(U)\geq 0$ for all $j,k$, $J_{jk}(U)=0$ for all $1\leq j\leq\operatorname{dim}S<k\leq n$, as well. 
Thus $J_U$ must be block diagonal, hence the orthogonal complement of $S$ is also invariant. This in turn implies that $\cV$ is a $*$-algebra~\cite[Thm.~11.5.1]{Gohberg06}, as claimed.
Note that this uses that $\cV$ contains the identity.
\end{proof}
Note that the converse is not true, as can be seen by considering the following example of the Bloch equations in Lindblad form:
Choosing the Lindblad terms $\sigma_+$, $2\sigma_-$, and $\sigma_z$, the system is clearly not unital but as the operators do not have a common eigenvector, the lattice of lazy subspaces is trivial and hence orthocomplemented.

The structure theorem for $*$-algebras~\cite[Thm.~5.6]{Farenick01} shows that up to a change of orthonormal basis, $*$-algebras have a particularly simple form: Let $\mc A\subseteq\mathbb C^{n\times n}$ be a $*$-algebra and let $m$ be the dimension of the center of $\mc A$. Then there exists $U\in\U(n)$ as well as positive integers $\{q_i\}_{i=1}^m$ and $\{k_i\}_{i=1}^m$ such that
\begin{align}\label{eq:star-algebra}
U\mc AU^* = \bigoplus_{i=1}^m \C^{q_i\times q_i}\otimes \mathds1_{k_i}\,.
\end{align}
Clearly it holds that $\sum_{i=1}^m q_i k_i = n$.
The vector of block-sizes in~\eqref{eq:star-algebra}
\begin{align*}
\tau=(\underbrace{q_1,\ldots,q_1}_{k_1},\ldots,\underbrace{q_m,\ldots,q_m}_{k_m})
\end{align*}
is of special importance as shown in the following lemma.
Either way in this case we say that $\mc A$ is \emph{of type $\tau$};
moreover, given $-L\in\wkl(n)$ unital we say that $-L$ is of type $\tau$ if the relaxation algebra $\cV$ is of type $\tau$.
To properly state the results in this section it is convenient to define the notion of refinement:

\begin{definition}[Refinement]
Let $n,k,l$ be positive integers and let $v\in\N^k$ and $w\in\N^l$ be vectors of positive integers such that their elements sum to $n$ respectively. We say that $v$ is a \emph{refinement} of $w$ if $Av=w$ for some $A\in\{0,1\}^{l\times k}$ with $\e^\top A=\e^\top$.
\end{definition}
Such matrices $A$ form a finite semigroup, and hence refinement is a preorder. Vectors differing only up to permutation are equivalent and we will not distinguish between them. Hence, in the following, we will think of $\tau$ as a multiset, usually represented in non-increasing order.

The following result is a direct consequence of the Krull--Schmidt Theorem, cf.~\cite[Prop.~3.2.5]{Hazewinkel04}.

\begin{lemma}
Let $\mc A$ be a $*$-algebra of type $\tau$. Then for any block diagonalization of $\mc A$ the vector of block sizes is refined by the vector $\tau$.
\end{lemma}

We will see below that the type of a unital system determines many of its control-theoretic properties.


\subsection{Stabilizability}

In contrast to general systems, where the exact shape of the stabilizable set $\stab_{\ref{eq:simplex-control-system}}$ is very difficult to describe, unital systems allow for a complete characterizations in algebraic terms using the relaxation algebra $\cV$, as we will show in Theorem~\ref{thm:stab-unital}. We begin with a simple result about stabilizable states in unital systems:

\begin{lemma} \label{lemma:unital-strongly-stab}
For unital systems, the set of stabilizable states equals the set of strongly stabilizable states.
\end{lemma}

\begin{proof}
Let $\lambda\in\Delta^{n-1}$ be stabilizable, i.e.~there exist $\ell\in\mathbb N$, $\mu_1,\ldots,\mu_\ell>0$, and $U_1,\ldots,U_\ell\in\SU(n)$ such that $-\sum_{k=1}^\ell\mu_k L_{U_k}\lambda=0$ and $\sum_{k=1}^\ell\mu_k=1$.
W.l.o.g.~there exists $k$ such that $L_{U_k}\neq 0$.
Then $\epsilon:=(\max_{j,k}|(L_{U_k})_{jj}|)^{-1}>0$ is well-defined and---similar to the proof of Lemma~\ref{lemma:triang-stab}---the matrix $\mathds 1-\epsilon L_{U_k}$ is doubly stochastic for all $k$;
here
we used unitality of the system together with Lemma~\ref{lemma:unital-systems}~\ref{it:unital-MU}.
This lets us compute
$
\lambda=\lambda-\epsilon(\sum_{k=1}^\ell\mu_k L_{U_k}\lambda)=\sum_{k=1}^\ell\mu_k (\mathds 1-\epsilon L_{U_k})\lambda
$,
meaning we expressed $\lambda$---which is an extreme point of the majorization polytope $\{\lambda'\in\Delta^{n-1}\,:\lambda'\prec\lambda\}$ \cite[Thm.~1]{Dahl10}---as a non-trivial convex combination of elements of $\{\lambda'\in\Delta^{n-1}\,:\lambda'\prec\lambda\}$ \cite[Ch.~2, Thm.~B.2]{MarshallOlkin}.
By definition of an extreme point this is only possible if $(\mathds 1-\epsilon L_{U_k})\lambda=\lambda$ for all $k$; hence $-L_{U_k}\lambda=0$ so $\lambda$ is strongly stabilizable.
\end{proof}

\begin{theorem} 
\label{thm:stab-unital}
Assume that $-L\in\wkl(n)$ is unital and let $\lambda\in\Delta^{n-1}$. 
Then $\lambda\in\stab_{\ref{eq:simplex-control-system}}$ if and only if 
the type of $-L$ is a refinement of the vector of multiplicities of $\lambda$.
\end{theorem}

\begin{proof}
``$\Leftarrow$'':
First assume that all Lindblad terms are in block-diagonal form such that the block sizes give a refinement of the multiplicities of $\lambda$. In particular $L_U$ will have the same block structure. Then there exists a permutation $\lambda'$ of $\lambda$ such that in the expression $L_U\lambda'$, all elements of $\lambda'$ which are multiplied with a given block of $L_U$ must have the same value. Thus Lemma~\ref{lemma:unital-systems}~\ref{it:unital-MU} shows $L_U\lambda'=0$ and hence $\lambda\in\stab_{\ref{eq:simplex-control-system}}$.

``$\Rightarrow$'':
Assume that no such block diagonal structure is achievable. 
Then the product $L_U\lambda$ will always have at least two distinct elements of $\lambda$ multiplied with a single block of $L_U$. 
Now consider a block where this happens. 
Then we argue that (at least one copy of) the smallest value of $\lambda$ falling into this block must have a strictly positive derivative.
Without loss of generality we assume that $L_U$ consists of a single block, and all copies of the smallest element of $\lambda$ are exactly the first $s$ elements of $\lambda$. 
Towards a contradiction assume that $\dot\lambda_i = 0$ for all $1\leq i\leq s$ (note that by unitality the smallest element of $\lambda$ cannot strictly decrease).
This means that $(L_U)_{ij}=0$ for all $i\leq s < j$. 
By Lemma~\ref{lemma:unital-systems}~\ref{it:unital-lambda} this implies that $(L_U)_{ji}=0$, contradicting the assumption that $L_U$ consists of a single block.
Hence $\lambda$ is not strongly stabilizable and by Lemma~\ref{lemma:unital-strongly-stab} it is not stabilizable at all.
\end{proof}

As a special case we can characterize unital systems which are stabilizable (that is, every state is stabilizable).

\begin{corollary}
\label{lemma:unital-stabilizable}
The following statements are equivalent.\smallskip
\begin{enumerate}[(i)]
\item \label{it:unital-stab-def} $-L$ is unital and $\stab_{\ref{eq:simplex-control-system}}=\Delta^{n-1}$.
\item \label{it:unital-stab-diag} The $U^*V_kU$ are simultaneously diagonal for some $U\in\SU(n)$ and for some (equivalently: every) choice of $\{V_k\}_{k=1}^r$. 
\item \label{it:unital-stab-normal} The $V_k$ are normal and commute for some (equivalently: every) choice of $\{V_k\}_{k=1}^r$.
\item \label{it:unital-stab-JU} $J(U)$ is diagonal (equivalently: $L_U=0$) for some $U\in\SU(n)$.\smallskip
\end{enumerate}
If $-L$ satisfies any---and hence all---of the above, we call it \emph{unital stabilizable}. 
\end{corollary} 

\begin{proof}
\ref{it:unital-stab-def} $\Leftrightarrow$ \ref{it:unital-stab-diag} is a consequence of Theorem~\ref{thm:stab-unital}.
\ref{it:unital-stab-diag} $\Leftrightarrow$ \ref{it:unital-stab-normal}:
\cite[Thm.~2.5.5]{HJ1}.
\ref{it:unital-stab-diag} $\Leftrightarrow$ \ref{it:unital-stab-JU} is trivial.
\end{proof}
Note that one could also use Theorem~\ref{thm:stab-system} instead of Theorem~\ref{thm:stab-unital} in the proof above.\smallskip

Due to Lemma~\ref{lemma:unital-strongly-stab}, in a unital stabilizable system every state is also strongly stabilizable, and any unitary $U$ which diagonalizes all Lindblad operators satisfies $L_U\lambda=0$ due to property~\ref{it:unital-stab-JU} above.
Hence using Proposition~\ref{prop:ham-stab-is-strongly-stab} one can find a corresponding compensating Hamiltonian for each regular $\lambda\in\Delta^{n-1}$. 
Finally let us characterize the following (trivial) case:

\begin{lemma} 
\label{lemma:purely-Hamiltonian}
The following are equivalent.\smallskip
\begin{enumerate}[(i)]
\item \label{it:pure-ham-def} $L=\iu\ad_H$ for some Hermitian matrix $H$.
\item \label{it:pure-ham-eig} All eigenvalues of $-L$ are on the imaginary axis.
\item \label{it:pure-ham-id} For some (equivalently: each) choice of Lindblad terms $\{V_k\}_{k=1}^r$ of $-L$, all $V_k$ are multiples of the identity.
\item \label{it:pure-ham-lazy} Every subspace is lazy for $-L$.
\item \label{it:pure-ham-JUs} $J(U)$ is diagonal (equivalently: $L_U=0$) for all $U\in\SU(n)$.\smallskip
\end{enumerate}
We call such systems \emph{purely Hamiltonian}.
\end{lemma}

\begin{proof}
\ref{it:pure-ham-def} $\Rightarrow$ \ref{it:pure-ham-eig} is obvious and the converse follows from~\cite[Thm.~18]{BNT08b}.
\ref{it:pure-ham-def} $\Leftrightarrow$ \ref{it:pure-ham-id}: This is clear from the fact that the relaxation algebra is well-defined, cf. Lemma~\ref{lemma:lindblad-alg-well-def}.
\ref{it:pure-ham-id} $\Leftrightarrow$ \ref{it:pure-ham-lazy}: Elementary.
\ref{it:pure-ham-id} $\Rightarrow$ \ref{it:pure-ham-JUs}: Immediate from the definitions.
\ref{it:pure-ham-JUs} $\Rightarrow$ \ref{it:pure-ham-lazy}: Follows from Lemma~\ref{lemma:LU-zeros}.
\end{proof}

It is clear that purely Hamiltonian systems are always unital stabilizable, which in turn are always unital. Moreover, these inclusions are strict. 


\subsection{Reachability}

Similar to stabilizability, reachability properties are also highly dependent on the structure of the relaxation algebra $\cV$. 
From Lemma~\ref{lemma:unital-systems}~\ref{it:unital-reach} we know that the reachable set is contained in the majorization polytope of the initial state.
Hence it is natural to ask when the reachable set is as large as it could be.
Due to the continuity of solutions, it is clear that in the reduced control system not every point in the majorization polytope can be reached. 
This however is an artifact of the Weyl symmetry, and one should really ask which states in the Weyl chamber of the initial state can be reached. 

\begin{theorem} 
\label{thm:reach-unital}
Assume that $-L\in\wkl(n)$ is unital and consider an initial state $\lambda\in\Delta^{n-1}_\down$ in the ordered Weyl chamber. 
Then exactly one of the following is the case: \smallskip
\begin{enumerate}[(i)]
\item \label{it:unital-reach-Ham} If $-L$ is purely Hamiltonian, then $\overline{\reach_{\ref{eq:relaxed-control-system}}(\lambda)} = \{\lambda\}$.
\item \label{it:unital-reach-all} If $-L$ is of type $(2,1,\ldots,1)$ or type $(1,\ldots,1)$ (i.e.~unital stabilizable) but not purely Hamiltonian, then
$$\overline{\reach_{\ref{eq:relaxed-control-system}}(\lambda)}\cap\Delta^{n-1}_\down 
= \{\mu\in\Delta^{n-1}_\down:\mu\prec\lambda\}
\,.
$$
\item \label{it:unital-reach-strict} \marginpar{to do: either prove or relax this statement and specify below what "cover all possibilites" means}
If $-L$ is of type greater than $(2,1,\ldots,1)$, then 
$$
\overline{\reach_{\ref{eq:relaxed-control-system}}(\lambda)}\cap\Delta^{n-1}_\down  \begin{cases}
=\{\e/n\}  &\text{if }\lambda=\e/n\\
\subsetneq \{\mu\in\Delta^{n-1}_\down:\mu\prec\lambda\} 
&\text{else}\,.
\end{cases}
$$
\end{enumerate}
\end{theorem}

\begin{proof}
It is clear that the three cases are mutually exclusive and cover all possible $-L$.
\ref{it:unital-reach-Ham}: This follows most easily from Lemma~\ref{lemma:purely-Hamiltonian}~\ref{it:pure-ham-JUs}.
\ref{it:unital-reach-all}: Let $U$ be a unitary such that $L_U$ is block diagonal with blocks of sizes $(2,1,1,\ldots,1)$ (this includes the $(1,1,\ldots,1)$ case). 
Because $L$ is not purely Hamiltonian, we can choose $L_U$ such that it is not diagonal, and thus the off-diagonal elements in the block of size $2$ are non-zero.
In addition, we have access to all $L_{U\pi}=\pi^\top L_U\pi$ where $\pi$ is any permutation matrix meaning we have (approximate) access to arbitrary two-level weight shifts (also called ``T-transforms''). 
By \cite[Ch.~2, Thm.~B.6]{MarshallOlkin} this implies the claim.

\ref{it:unital-reach-strict}:
W.l.o.g.~$\lambda\neq\e/n$.
By assumption, $M_U$ must always have either one block of size at least $3$ or at least two blocks of size at least $2$. 
This means that at least $3$ elements of $\lambda$ will be acted on by $M_U$. 
In particular, as in the proof of Theorem~\ref{thm:stab-unital}, the largest (smallest) of these elements will have a strictly negative (positive) derivative. 
Since the edges of the majorization polytope incident to $\lambda$ correspond to the mixing of two neighboring elements in $\lambda$, no derivative pointing along an edge can be achieved. 
%
%
Let $v$ be a unit vector pointing along one of the edges. Let $\alpha$ be a linear functional with $\alpha(v)=0$ such that $\alpha$ is strictly positive on any point in the majorization polytope that does not lie on the edge defined by $v$. Consider the function
$
f_{U,\alpha}(\lambda) = -\alpha(L_U\lambda)
$
on a compact neighborhood $V$ of $\lambda$ which is disjoint from $\stab_{\ref{eq:simplex-control-system}}$. 
One can show that $f^-(\lambda) = \min_U f_{U,\alpha}(\lambda)$ is continuous on $V$. By shrinking $V$ (while keeping $\lambda$ inside) we can assume that $f^-(\lambda)\geq\epsilon$ on $V$ for some $\epsilon>0$.
Let $\mu(t)$ be any solution to~\eqref{eq:relaxed-control-system} starting at $\lambda$, and let $T>0$ be such that $\mu(t)\in V$ for all $t\in[0,T]$. Then
$$
\alpha(\mu(T)) 
= \int_0^T \partial_t \alpha(\mu(t))dt
= \int_0^T \alpha (\partial_t \mu(t))dt
\geq T\epsilon > 0\,.
$$
Thus no point on an edge incident to $\lambda$ which is not in $V$ can be in $\overline{\reach_{\ref{eq:relaxed-control-system}}(\lambda)}$.
\end{proof}
This recovers the main result of~\cite{Styliaris19} where it was shown, using similar arguments, and based on~\cite{rooney2018}, that a generator $-L$ can effect all state transfers respecting majorization if and only if we are in case~\ref{it:unital-reach-all}. 
Moreover, the case of a single Lindblad term $V_k$ which is not a multiple of the identity (again covered by case~\ref{it:unital-reach-all}) continues to hold in the infinite dimensional case~\cite{OSID19}.

Note that the results of Theorem~\ref{thm:reach-unital} immediately lift to the full control system~\eqref{eq:bilinear-control-system} via Proposition~\ref{prop:reach-equiv} and using the fact that $\rho\prec\sigma$ if and only if $\spec^\down(\rho)\prec\spec^\down(\sigma)$.

\begin{corollary}
For $-L$ unital, the reduced control system is reverse coolable if and only if case~\ref{it:unital-reach-all} of Theorem~\ref{thm:reach-unital} is satisfied.
\end{corollary}

\subsection{Accessibility} \label{sec:unital-access}

In our discussion of accessibility in Section~\ref{sec:accessible} we forwent the unital case, which we will treat now.

\begin{proposition}\label{prop:unital-dir-access}
Let $-L\in\wkl(n)$ unital be given.
If $n=2$, the reduced control system~\eqref{eq:simplex-control-system} is generically (cf.~footnote~\ref{footnote:generic})
directly accessible if and only if $-L$ is not purely Hamiltonian.
If $n>2$, then exactly one of the following holds.\smallskip
\begin{enumerate}
\item The reduced control system~\eqref{eq:simplex-control-system} is generically directly accessible.
\item The reduced control system~\eqref{eq:simplex-control-system} is nowhere directly accessible and each $J(U)$ has identical off-diagonal elements.
\end{enumerate}
\end{proposition}

\begin{proof}
If $n=2$ and $-L$ is not purely Hamiltonian, then by Lemma~\ref{lemma:purely-Hamiltonian} there is some $U\in\SU(2)$ such that $L_U$ is not identically $0$. 
But then $L_U$ vanishes at only one point on $\Delta^1$, and hence reduced control system is generically directly accessible.

Now let $n>2$.
We consider accessibility at the vertices $e_i$ of the simplex.
If there is some $i\in\{1,\ldots,n\}$ and some $U\in\SU(n)$ such that $-L_Ue_i$ is not invariant under all permutations stabilizing $e_i$, then the system is accessible at $e_i$ and hence using~\cite[Prop.~4.14]{MDES23} it is generically directly accessible.
If no such $e_i$ and $U$ exist, then all columns of all $J(U)$ are constant on their off-diagonal elements.
It is easy to see that this implies that all off-diagonal elements of all $J(U)$ are the same.
In this case $\mf L$ is one dimensional, and hence the system is nowhere directly accessible.
\end{proof}
As in the non-unital case, if the reduced system~\eqref{eq:simplex-control-system} is generically directly accessible, then by~\cite[Prop.~4.15]{MDES23} the full bilinear system~\eqref{eq:bilinear-control-system} is generically accessible.

A sufficient condition for being nowhere directly accessible is the following:

\begin{corollary}
If the non-Hamiltonian part of $-L$ is unitarily invariant, then $\mf L$ contains a single vector field, which is necessarily permutation invariant. If, in addition, $n>2$, then the system is nowhere directly accessible.
\end{corollary}

\section{Conclusions}

We have established a rigorous correspondence principle between the dynamics of the reduced control system~\eqref{eq:simplex-control-system} and its parent, 
the full bilinear control system~\eqref{eq:bilinear-control-system} based on interrelating their solutions to the respective equations 
of motion in the Equivalence Theorem~\ref{thm:equivalence}.

For key quantum engineering questions we have given answers for both systems~\eqref{eq:simplex-control-system} 
and~\eqref{eq:bilinear-control-system}:\smallskip
\begin{enumerate}[(1)]
\item {\em Stabilizability:} Theorem~\ref{thm:stab-system} characterizes stabilizability of the reduced system in terms of simultaneously triangularizable Lindblad terms $\{V_k\}$---and via Lie's Theorem to solvable Lie algebras.
On the full system~\eqref{eq:bilinear-control-system} this implies approximate viability of every $\SU(n)$-orbit.
\item Generic {\em Accessibility} of~\eqref{eq:bilinear-control-system} is settled for non-unital systems in Prop.~\ref{prop:accessibility}, and for unital systems in Prop.~\ref{prop:unital-dir-access}.
\item Asymptotic {\em Coolability:} Theorem~\ref{thm:asymptotic-coolability} relates coolability of 
both~\eqref{eq:simplex-control-system} and~\eqref{eq:bilinear-control-system} to common right eigenvectors 
of the $\{V_k\}_{k=1}^r$ that are not common left eigenvectors. If a system is coolable, then a time-independent Hamiltonian is sufficient.
\item Approximate {\em Controllability:} Proposition~\ref{prop:approx-ctrl} relates approximate controllability for~\eqref{eq:simplex-control-system} and~\eqref{eq:bilinear-control-system} to coolability and reverse coolability in~\eqref{eq:simplex-control-system}.
\item {\em Reachability:} Theorem~\ref{thm:reach-unital} settles approximate reachability of all majorized states in unital systems for~\eqref{eq:simplex-control-system} and~\eqref{eq:bilinear-control-system}.
\end{enumerate}

\smallskip
More details of what can be concluded from the dynamics of the reduced system~\eqref{eq:simplex-control-system} for the full system~\eqref{eq:bilinear-control-system} are summarized in Table~\ref{tab:conclusion}.
\color{black}

\begin{sidewaystable}
\centering
\def\arraystretch{1.3}
\caption{The Equivalence Theorem~\ref{thm:equivalence} shows how 
the full control system~\eqref{eq:bilinear-control-system} and the reduced control system~\eqref{eq:simplex-control-system} are equivalent. 
The table summarizes the equivalence for 
several important control-theoretic notions giving the corresponding 
algebraic conditions. 
}\vspace{-2mm}
\label{tab:conclusion}
\begin{tabular}{llll}
\hline\hline\\[-2.5mm]
Reduced control system~\eqref{eq:simplex-control-system} & Algebraic condition & Full control system~\eqref{eq:bilinear-control-system} & Reference \\[2.5mm]
\hline\hline\\[-2.5mm]
(regular) strongly stabilizable point & kernel of some $L_U$ & \makecell[tl]{(regular) strongly stabilizable\\ point in orbit} & Prop.~\ref{prop:ham-stab-is-strongly-stab}\\[5.0mm]
viable face of dimension $d-1$ & lazy subspace of dimension $d$ & \makecell[tl]{collecting subspace via\\ Hamiltonian control} & Prop.~\ref{prop:viable-faces} \\[5.0mm]
stabilizable system & simultaneously triangularizable & $\SU(n)$-orbits approximately viable & Thm.~\ref{thm:stab-system} \\[5.0mm]
generically directly accessible & non-unital or accessible at vertex & generically accessible & \makecell[tl]{Prop.~\ref{prop:accessibility} \&\\ Prop.~\ref{prop:unital-dir-access}} \\[5.0mm]
asymptotically coolable & common right but not left eigenvector & \makecell[tl]{asymptotically coolable with \\ time-independent Hamiltonian} & Thm.~\ref{thm:asymptotic-coolability} \\[5.0mm]
directly reachable face & \makecell[tl]{flow ratios, \\collecting but not enclosing} & time-independent Hamiltonian & \makecell[tl]{Lem.~\ref{lemma:ham-reachable} \&\\ Prop.~\ref{prop:directly-reachable-faces}}\\[5.0mm]
reverse coolable & lazy subspaces of each dimension & approx. reachability from pure state & Prop.~\ref{prop:rev-coolable} \\[3.0mm]
approx. controllable & --- & approx. controllable & Prop.~\ref{prop:approx-ctrl} \\[3.0mm]
\hline\\[-2.5mm]
unital stabilizable point & relaxation $*$-algebra and refinement & \makecell[tl]{(regular) strongly stabilizable\\ point in orbit} & Thm.~\ref{thm:stab-unital} \\[2.0mm]
\makecell[tl]{unital reachability of \\majorization polytope} & type of relaxation $*$-algebra & \makecell[tl]{approx. reachability of\\ majorized states} & Thm.~\ref{thm:reach-unital} \\[7.0mm]
\hline\hline\\[-4.5mm]
\multicolumn{4}{l}{NB: the entries in a row are not always exactly equivalent, but some conditions might only be necessary or sufficient. For precise}\\[-2.5mm]
\multicolumn{4}{l}{statements we refer the the relevant results.}
\end{tabular}
\end{sidewaystable}

\section*{Acknowledgments}
The project was funded i.a.~by the Excellence Network of Bavaria under ExQM, by {\it Munich Quantum Valley} of the Bavarian State Government with funds from Hightech Agenda {\it Bayern Plus} (E.M., F.v.E. and T.S.H.), as well as the Einstein Foundation (Einstein Research Unit on quantum devices) and the {\sc math}+ Cluster of Excellence (F.v.E.).

\appendix

\section{Proof of the Equivalence Theorem~\ref{thm:equivalence}} \label{app:equivalence}
The proof of Theorem~\ref{thm:equivalence} has been carried out in greater generality in~\cite{MDES23}.
Here we only show how the bilinear control system~\eqref{eq:bilinear-control-system} can be reinterpreted in the setting of semisimple orthogonal symmetric Lie algebras. 
The Lie algebra in question is $\mf{sl}(n,\C)=\mf{su}(n)\oplus\mf{herm}_0(n,\C)$, where $\mf{herm}_0(n,\C)$ denotes the traceless Hermitian matrices. 
Since for density matrices $\rho$ it holds that $\tr(\rho)=1$, we will consider the shifted operator $\bar\rho=\rho-\mathds1/n$ which satisfies $\bar\rho\in\mf{herm}_0(n,\C)$.
The corresponding shifted Kossakowski--Lindblad generator has the form $\bar L(\bar\rho):=L(\rho)=L(\bar\rho)+L(\mathds1)/n$, which is affine linear. 
Since $\Ad_U(\bar\rho)=\overline{\Ad_U(\rho)}$ and $\ad_H(\rho)=\ad_H(\bar\rho)$, the ``shifted'' control system
$
\dot{\bar\rho} = -(\ad_H+\bar L)(\bar\rho)
$
and~\eqref{eq:bilinear-control-system} are state space equivalent\footnote{Two control systems are state space equivalent if there is a smooth diffeomorphism between their state spaces which also maps the drift and control vector fields of one system to the other.}. 
After this transformation, the new system is in the form considered in~\cite{MDES23}, and hence we can define the equivalent reduced control system. 
The reduced state space is $\mf{diag}_0(n,\R)\cong\R^n_0$,
and the induced vector fields are
$
-\bar L_U \bar\lambda = -\Pi_{\diag}\circ\Ad_U^{-1}\circ \bar L \circ \Ad_U\circ\diag(\bar\lambda).
$
Applying the definitions one finds that $L_U\lambda=\bar L_U\bar\lambda$. 
This shows that the ``shifted'' reduced control system is state space equivalent to the reduced control system~\eqref{eq:simplex-control-system}. 
This (linear) state space equivalence shows that show that~\cite[Thms.~3.8 \& 3.14]{MDES23} imply the Equivalence Theorem~\ref{thm:equivalence}. More generally, all results from~\cite{MDES23} presupposing an affine linear drift term can be applied to our system.

\section{Some Technical Results} \label{app:comps}
Here we collect some useful properties of the matrices $J(U)$ and $L_U$.

\begin{lemma} \label{lemma:LU-zeros}
Given arbitrary $\{V_k\}_{k=1}^r\subset\mathbb C^{n\times n}$ as well as any $i,j\in\{1,\ldots,n\}$ with $i\neq j$, the following statements hold.\smallskip
\begin{enumerate}[(i)]
\item \label{it:zeros-1} Given $\alpha_1,\ldots,\alpha_m>0$ with $\sum_{k=1}^m\alpha_k=1$ as well as $U_1,\ldots,U_m\in\SU(n)$
define $M:=\sum_{k=1}\alpha_kL_{U_k}$. If $M_{ij}=0$, then $(L_{U_k})_{ij}=0$ for all $k=1,\ldots,m$.
\item \label{it:zeros-2} If $(L_U)_{ij}=0$ for some $U\in\SU(n)$, then $(U^*V_kU)_{ij}=0$ for all $k=1,\ldots,r$.
\end{enumerate}
\end{lemma}

\begin{proof}
\ref{it:zeros-1}: As $-L_{U_k}$ is the generator of a stochastic matrix
we know $(-L_{U_k})_{ij}\geq 0$ because $i\neq j$. Thus
$(-M)_{ij}=0$ implies $(-\alpha_kL_{U_k})_{ij}= 0$ for all $k$, hence $(L_{U_k})_{ij}=0$.
\ref{it:zeros-2}: Because $i\neq j$ we know that $(-L_U)_{ij}=(J(U))_{ij}$
so the former being zero forces $(U^*V_kU)_{ij}=0$ for all $k$.
\end{proof} 

\begin{lemma} \label{lemma:MU-row-column-sums}
Let $\{V_k\}_{k=1}^r$ be a family of $n$-dimensional Lindblad terms and let $U\in\U(n)$ be arbitrary.
The row and column sums of $J(U)$ (as defined in~\eqref{eq:def_JU}) are
$$
(J(U) \e)_i = \sum_{k=1}^r \langle u_i| V_k V_k^* |u_i\rangle, \quad  
(J(U)^\top \e)_j =\sum_{k=1}^r \langle u_j|V_k^* V_k |u_j\rangle 
$$
respectively, where $|u_i\rangle = U|i\rangle$. 
In particular if $\prec$ denotes standard majorization, then $J(U)\e\prec\spec^\down(\sum_{k=1}^r V_k V_k^*)$ and $J(U)^\top\e \prec\spec^\down(\sum_{k=1}^r V_k^*V_k)$. 
It follows that
\begin{equation*}
\begin{split}
((J(U)-J(U)^\top)\e)_i&=\sum_{k=1}^r \langle u_i| [V_k,V_k^*] |u_i\rangle\,,\\
((J(U)+J(U)^\top)\e)_i&=\sum_{k=1}^r \langle u_i| \{V_k,V_k^*\} |u_i\rangle\,,
\end{split}
\end{equation*}
where $\{\cdot,\cdot\}$ denotes the anti-commutator.
So $(J(U)-J(U)^\top)\e\prec \spec^\down(\sum_{k=1}^r [V_k,V_k^*])$ and $(J(U)+J(U)^\top)\e\prec\spec^\down(\sum_{k=1}^r\{V_k,V_k^*\})$.
\end{lemma}

\begin{proof}
We compute the $i^\text{th}$ row sum
$$
(J(U)\e)_i
= \sum_{j=1}^n \langle i|J(U) |j\rangle
= \sum_{j,k=1}^{n,r} |\langle u_i| V_k |u_j\rangle|^2
= \sum_{k=1}^r \langle u_i| V_k V_k^* |u_i\rangle.
$$
The computation for the $j^\text{th}$ column is analogous, and the other claims follow immediately
using the Schur--Horn theorem \cite{Schur23,Horn54}.
\end{proof}

\section{Structure of the Lindblad Equation} \label{app:lindblad-eq}
In this section we recall some results pertaining to the structure theory of the Lindblad equation based primarily on~\cite{BNT08b}, taking care to present them in their proper mathematical context.
Since many of these results can also be found in other sources, such as~\cite{Kraus08,TV09,Schirmer10}, often with significantly differing terminology, we give precise definitions of all concepts used in this paper. 
See also~\cite{Deschamps16} and references therein for related results (obtained in the Heisenberg picture) which extend into the infinite-dimensional setting.
We introduce several matrix algebras generated by the Lindblad terms and
consider their invariant subspace lattices.
Once this structure is understood, we can modify the Lindblad equation to achieve certain dynamical properties.


The celebrated result by Gorini, Kossakowski, and Sudarshan~\cite{GKS76},
and Lindblad~\cite{Lindblad76}
establishes the form of the generators of so-called \emph{quantum-dynamical semigroups} (i.e.~of continuous maps $t\mapsto\Phi_t$ on $\mathbb R_+$ into the completely positive trace-preserving linear maps, which satisfy $\Phi_0={\mathds1}$ and $\Phi_t\circ\Phi_s=\Phi_{t+s}$ for all $t,s\geq0$).
\begin{remark}
For our purposes \textnormal{completely positive dynamical semigroups} (that is, quantum-dynamical semigroups without the trace preservation condition) and their generators 
turn out to be more convenient.
Hence we will formulate many results for these more general objects.
\end{remark}
First, generators of completely positive dynamical semigroups can be characterized as follows, cf.~\cite[Thm.~3]{Lindblad76}:

\begin{lemma} 
A linear map $-L\in\mc L(\C^{n\times n})$ is the generator of a completely positive dynamical semigroup if and only if there is some completely positive\footnote{Recall that a map $\phi$ is completely positive if and only if it can be written in the form $\sum_k V_k(\cdot)V_k^*$, and any such operators $V_k$ are called Kraus operators.} linear map $\phi=\sum_{k=1}^r V_k(\cdot)V_k^*$ and some $K\in\C^{n\times n}$ such that 
\begin{align} \label{eq:qds}
-L(\rho) = 
\phi(\rho)
 - K\rho - \rho K^*.
\end{align}
Moreover, $t\mapsto e^{-tL}$ is a quantum-dynamical semigroup if and only if $-L$ is additionally trace-annihilating, that is,~\eqref{eq:qds} holds and
there exists $H_0\in\mathbb C^{n\times n}$ Hermitian such that
\begin{align} \label{eq:lindblad2}
K = \iu H_0 + \frac12\sum_{k=1}^rV_k^*V_k\,.
\end{align}
\end{lemma}
Note that condition~\eqref{eq:qds} is also known as \textit{conditional complete positivity} of $-L$ \cite[Thm.~14.7]{EL77}.
Either way, plugging~\eqref{eq:lindblad2} into~\eqref{eq:qds} recovers the Lindblad equation given in the introduction~\eqref{eq:lindblad-equation}.\smallskip

It is important to note that the choice of Hamiltonian $H_0$ and of Lindblad terms $V_k$ is not unique. The following well-known result characterizes the freedom of representation of a given Kossakowski--Lindblad generator, cf.~\cite[Eq.~(3.72) \& (3.73)]{BreuPetr02}.

\begin{lemma}
\label{lemma:freedom-of-reps}
Let $\{K,V_k:k=1,\ldots,r\}$ be a representation of some generator of a dynamical semigroup of completely positive maps. Given $\{V_k'\}_{k=1}^s$, if it holds that
\begin{align} \label{eq:unitary-shuffling}
V_k' = \sum_{j=1}^{r} u_{kj}V_j \quad\text{ for all }\, k=1,\ldots, s
\end{align}
for some $U:=(u_{kj})_{k,j=1}^{s,r}\in\mathbb C^{s\times r}$ which
satisfies either $UU^*=\mathds1_s$ (if $r\geq s$) or $U^*U=\mathds1_r$ (if $r\leq s$),
then $\{K,V_k:k=1,\ldots,r\}$ and $\{K,V'_k:k=1,\ldots,s\}$ are equivalent. Similarly, if $\{V_k'\}_{k=1}^s$ and $K'$ satisfy, for some $c_k\in\C$ and $\lambda\in\R$, that
\begin{equation} \label{eq:scalar-shift}
\begin{split}
V_k'&=V_k+c_k\mathds1\\
K'&=K+\sum_{k=1}^r\overline{c}_kV_k+\left(\iu\lambda + \frac{1}{2}\sum_{k=1}^r|c_k|^2\right)\mathds1\,,
\end{split}
\end{equation}
then $\{K,V_k:k=1,\ldots,r\}$ and $\{K',V'_k:k=1,\ldots,s\}$ are equivalent. Conversely, if two representations $\{K,V_k:k=1,\ldots V_r\}$ and $\{K',V_k':k=1,\ldots V_s\}$ define the same generator,
then they are related by some sequence of the transformations above.
\end{lemma}
\begin{proof}
For the reader's convenience we provide a short proof.
The first statement follows from~\eqref{eq:qds} together with the characterization of ``uniqueness'' of Kraus operators \cite[Coro.~2.23]{Watrous18}.
Next, \eqref{eq:scalar-shift} is a straightforward computation.
Finally, if $\{K,V_k:k=1,\ldots V_r\}$ and $\{K',V_k':k=1,\ldots V_s\}$ define the 
same generator, then shifting $V_k\to V_k-\frac{{\rm tr}(V_k)}{n}\mathds 1$, 
$V_k'\to V_k'-\frac{{\rm tr}(V_k')}{n}\mathds 1$ turns $K,K'$ into $K_1,K_1'$ 
(according to~\eqref{eq:scalar-shift}), respectively. 
Thus we may assume that all $V_k,V_k'$ are traceless.
In this case, vectorizing \cite[Ch.~2.4]{MN07} $\sum_k V_k(\cdot)V_k^*-K_1(\cdot)-(\cdot)K_1^*=\sum_k V_k'(\cdot)(V_k')^*-K_1'(\cdot)-(\cdot)(K_1')^*$ yields
$
\sum_k \overline{V_k}\otimes V_k-\mathds1\otimes K_1-\overline{K_1}\otimes\mathds 1=\sum_k \overline{V_k'}\otimes V_k'-\mathds1\otimes K_1'-\overline{K_1'}\otimes\mathds 1\,.
$
Using ${\rm tr}(V_k)={\rm tr}(V_k')=0$ for all $k$, taking the partial trace over the first system shows
$
\overline{{\rm tr}(K_1-K_1')}\frac{\mathds 1}n=-(K_1-K_1')
$.
Taking the trace again yields ${{\rm tr}(K_1-K_1')}\in i\mathbb R$, meaning there exists $\lambda\in\mathbb R$ such that $K_1'=K_1+i\lambda\mathds 1$ (in accordance with~\eqref{eq:scalar-shift}).
This also means that the shifted $V_k,V_k'$ satisfy $\sum_k V_k(\cdot)V_k^*=\sum_k V_k'(\cdot)(V_k')^*$ which---again by
\cite[Coro.~2.23]{Watrous18}---shows that there exists an isometry $U$ relating the two sets via~\eqref{eq:unitary-shuffling}. This concludes the proof.
\end{proof}
In particular this result shows that the Lindblad terms can always be chosen traceless.
Indeed, this is one way to ensure that Kossakowski--Lindblad generators decompose uniquely: if ${\rm tr}(V_j)=0$ for all $j$, then there exists unique $H_0\in\mathfrak{su}(n)$ such that $-L=-\iu\ad_{H_0} -\sum_{k=1}^r\Gamma_{V_k}$ and if  for all $j$, cf.~\cite[Thm.~2.2]{GKS76}, \cite{Davies80unique}.

\begin{remark}
For the Lindblad equation, the transformation of $H_0$ resulting from the shift~\eqref{eq:scalar-shift} is given by
$
H' = H_0 + \frac{\iu}{2}\sum_{k=1}^r (c_k^*V_k-c_kV_k^*)  
+ \lambda \mathds1 
$.
\end{remark}

The central objects considered in studying the structure of the Lindblad equation are the following matrix algebras:

\begin{definition} \label{def:lindblad-algebras}
Let $-L\in\mathcal L(\mathbb C^{n\times n})$ be the generator of a completely positive dynamical semigroup represented by $\{K,V_k:k=1,\ldots,r\}$. We define the following unital subalgebras\footnote{
A unital matrix algebra is one which contains the identity matrix.
Since all our matrix algebras will be unital, we will omit this term in the following to avoid confusion.
}
of $\mathbb C^{n\times n}$:
$$
\cV = \generate{\mathds1,V_k:k=1,\ldots,r}{alg},
\qquad
\cV^+ = \generate{\mathds1,K,V_k:k=1,\ldots,r}{alg},
$$
called the \emph{relaxation algebra} and \emph{extended relaxation algebra},
respectively.
\end{definition}
Crucially, as a direct consequence of Lemma~\ref{lemma:freedom-of-reps} these algebras are well-defined:

\begin{lemma} \label{lemma:lindblad-alg-well-def}
The algebras $\cV$ and $\cV^+$ are well-defined, meaning that they only depend on $-L$, and not on the chosen representation.
\end{lemma}
%
\noindent Note, however, that\smallskip
\begin{itemize}
\item the algebras are not sufficient to determine the generator. 
\item for Lemma~\ref{lemma:lindblad-alg-well-def} to hold it is necessary to include the identity in the definition of the (extended) relaxation algebra. \smallskip
\end{itemize}

Below we will be interested in invariant subspaces of the matrix algebras defined above, so let us take a brief moment to recall some basic facts. 
A subspace $S\subseteq\mathbb C^n$ is \emph{invariant} for a set of matrices $\mc A$ if $AS\subseteq S$ for all $A\in\mc A$. 
If $P_S$ is a projection
onto $S$, then $S$ is invariant for $\mc A$ if and only if $(\mathds1-P_S)AP_S=0$ for all $A\in\mc A$.
Moreover, the set of all invariant subspaces of $\mc A$ forms a lattice, meaning that if $S$ and $T$ are invariant, then so are $S\cap T$ and $S+T$.
Since the invariant subspaces of $\mc A$ and $\generate{\mc A}{alg}$ coincide, we will work with the (well-defined) algebras $\cV$ and $\cV^+$.
For some mathematical background we refer to~\cite{Farenick01,Radjavi00} for matrix algebras, to~\cite{Gohberg06} for invariant subspaces, and to~\cite{Birkhoff67} for lattice theory.

The following lemma---which is an immediate consequence of~\eqref{eq:qds}---is quite useful in relating the algebras defined above to the structure of dynamical semigroups.

\begin{lemma} \label{lemma:qds-projectors}
Let $-L\in\mathcal L(\mathbb C^{n\times n})$ be the generator of a completely positive dynamical semigroup represented by $\{K,V_k:k=1,\ldots r\}$.
Let $O,P,Q,R\in\mathbb C^{n\times n}$ be orthogonal projections. Then for all $A\in\mathbb C^{n\times n}$ it holds that
\begin{align*}
O\;L(PAR)\;Q=\sum_{k=1}^r (OV_kP) A (QV_kR)^* - (OKP)A(RQ) - (OP)A(QKR)^*
\end{align*}
\end{lemma}


\begin{lemma} \label{lemma:lazy}
Let $-L\in\mathcal L(\mathbb C^{n\times n})$ be the generator of a completely positive dynamical semigroup represented by $\{K,V_k:k=1,\ldots,r\}$.
Given any $U\in\U(n)$ and any subspace $\{0\}\neq S\subseteq\mathbb C^n$---where $P_S$ denotes the orthogonal projection onto $S$---the following are equivalent: \smallskip
\begin{enumerate}[(i)]
\item \label{it:lazy-inv-Vk} $S$ is a common invariant subspace of all $V_k$.
\item \label{it:lazy-inv-V} $S$ is an invariant subspace of $\cV$.
\item \label{it:lazy-all} $P^\perp_S L(P_SA P_S)P^\perp_S=0$ for all $A\in\mathbb C^{n\times n}$.
\item \label{it:lazy-one} $P^\perp_S L(\rho)P^\perp_S=0$ for some state $\rho$ with $\supp(\rho)=S$.
\item \label{it:unitary} If the first $\operatorname{dim}S$ columns of $U^*$ span $S$, then $-L_U$ is block triangular, i.e.~for all $1\leq j\leq \operatorname{dim}S<i\leq n$ one has $(-L_U)_{ij}=0$.
\end{enumerate} \smallskip
If one and thus all of these conditions hold,
we say that $S$ is a \emph{lazy subspace} (of $-L$). 
\end{lemma}

\begin{proof} 
\ref{it:lazy-inv-Vk} $\Leftrightarrow$ \ref{it:lazy-inv-V}:
Note that by Lemma~\ref{lemma:lindblad-alg-well-def} the invariant subspaces of the $V_k$ are well-defined and exactly the invariant subspaces of $\cV$. 
\ref{it:lazy-inv-Vk} $\Rightarrow$ \ref{it:lazy-all}:
By Lemma~\ref{lemma:qds-projectors},
$P_S^\perp L(P_SAP_S)P_S^\perp=\sum_{k=1}^r (P_S^\perp V_kP_S) A (P_S^\perp V_kP_S)^*$.
\ref{it:lazy-all} $\Rightarrow$ \ref{it:lazy-one}: Trivial. 
\ref{it:lazy-one} $\Rightarrow$ \ref{it:lazy-inv-Vk}:
We compute
$
0=P_S^\perp L(P_S\rho P_S)P_S^\perp=\sum_{k=1}^r (P_S^\perp V_kP_S\sqrt{\rho}) (P_S^\perp V_kP_S\sqrt{\rho})^*$,
which shows $P_S^\perp V_kP_S\sqrt{\rho}=0$ for all $k$. Since $\supp(\rho)=\supp(\sqrt\rho)=S$, this forces $P_S^\perp V_kP_S=0$ for all $k$.
\ref{it:lazy-inv-Vk} $\Leftrightarrow$ \ref{it:unitary}: This is due to $(-L_U)_{ij}=\sum_{k=1}^r|\langle Ui|V_k|U_j\rangle|^2$ as follows from~\eqref{eq:J-comp}.
\end{proof}

Intuitively, a lazy subspace is one which is invariant under $-L$ to first order, hence the name. Importantly, they do not depend on $K$ (or $H_0$), and thus they are independent of the control. A stronger notion of invariance is the following (cf.~\cite{Shabani05,BNT08b}), which is an immediate consequence of \cite[Lem.~11]{BNT08b}:

\begin{lemma} \label{lemma:collecting}
Let $-L\in\mathcal L(\mathbb C^{n\times n})$ be the generator of a completely positive dynamical semigroup represented by $\{K,V_k:k=1,\ldots,r\}$.
Given any subspace $\{0\}\neq S\subseteq\mathbb C^n$ where $P_S$ denotes the orthogonal projection onto $S$,
the following are equivalent:\smallskip
\begin{enumerate}[(i)]
\item \label{it:invariant-etL} $P_S\mf{pos}_1(n)P_S$ is invariant under $e^{-tL}$ for all $t\geq0$, that is, for all $\rho\in \mf{pos}_1(n)$
it holds that $P_Se^{-tL}(P_S\rho P_S)P_S=e^{-tL}(P_S\rho P_S)$.
\item \label{it:invariant-L} $P_S L(P_S\rho P_S)P_S = L(P_S\rho P_S)$.
\item \label{it:invariant-algebra} $S$ is an invariant subspace for $\cV^+$.\smallskip
\end{enumerate}
If one and thus all of these conditions hold,
we say that $S$ is a \emph{collecting subspace}. 
\end{lemma}

%
Collecting subspaces are closely related to fixed points, see~\cite[Prop.~5 \&~Lem.~12]{BNT08b}.

\begin{lemma} \label{lemma:conf-fixed}
The following hold: \smallskip
\begin{enumerate}[(i)]
\item If $\rho\in\mf{pos}_1(n)$ is a fixed point, then $\supp(\rho)$ is collecting.
\item If $S\subseteq\mathbb C^n$ is collecting, then there is a fixed point $\rho$ with $\supp(\rho)\subseteq S$.
\end{enumerate}
\end{lemma}

Lemmas~\ref{lemma:lazy} and~\ref{lemma:collecting}, together with the preceding discussion, show that the lazy subspaces form a lattice of invariant subspaces, and that the collecting subspaces form a sublattice thereof\,\footnote{
To be precise, we only obtain lattices after including the trivial subspace $\{0\}$ which we excluded from both definitions.}.
A simple but important consequence is that lazy subspaces are exactly those which can be made collecting using a time-independent Hamiltonian:

\begin{corollary} \label{coro:stab-ham}
Let $-L\in\wkl(n)$ be a Kossakowski--Lindblad generator.
If $S\subseteq\mathbb C^n$ is a lazy subspace, then there exists a Hamiltonian $H_S$ such that $S$ is collecting for $-L+\iu\ad_{H_S}$.
Indeed, any Hamiltonian satisfying 
$$
P_S^\perp H_S P_S
=
-P_S^\perp \Big(H_0+\tfrac1{2\iu}\sum_{k=1}^rV_k^*V_k\Big) P_S
$$
will do the job. We call $H_S$ a \emph{stabilizing Hamiltonian} for $S$.
\end{corollary}
It is clear that such $H_S$ always exists since, in an appropriate basis, the condition only determines matrix elements below the diagonal.

As a consequence of Corollary~\ref{coro:stab-ham} and~\cite[Thm.~13]{BNT08b} we have the following relation between minimal collecting subspaces and extremal fixed points.

\begin{lemma} \label{lemma:min-collecting}
Let $-L\in\wkl(n)$ be a Kossakowski--Lindblad generator and let $S\subseteq\mc H$ be a subspace.
The following are equivalent: \smallskip
\begin{enumerate}[(i)]
\item \label{it:minimal-collecting} $S$ is a minimal\,\footnote{A minimal collecting subspace is one which does not contain a smaller collecting subspace. It is an atom in the lattice of collecting subspaces.} collecting subspace of $S$.
\item \label{it:extremal-fixed} There is an extremal\,\footnote{An extremal fixed point is one which is not a non-trivial convex combination of two distinct fixed points, that is, it is an extreme point of the convex set of fixed points.} fixed point with support equal to $S$.
\item \label{it:unique-fixed} $-L$ has a unique fixed point on $S$ and it has full rank (on $S$).
\end{enumerate}
\end{lemma}
This result recovers~\cite[Thm.~2]{Kraus08} and~\cite[Prop.~2]{Schirmer10}. \smallskip


We say that a subspace $S$ is \emph{decaying} if $\tr(P_S e^{-Lt}\rho)\to0$ as $t\to\infty$ for all states $\rho\in\mf{pos}_1(n)$. 
The complement of a decaying subspace is called an \emph{asymptotic} subspace.
Note that the property of being decaying is preserved under taking sums\footnote{This follows from the fact that if $\tr(P_S\rho)=0$ and $\tr(P_T\rho)=0$, then $\supp(\rho)\subseteq S^\perp\cap T^\perp=(S+T)^\perp$ and hence $\tr(P_{S+T}\rho)=0$.} and subspaces.
Hence there exists a unique \emph{maximal decaying subspace}, and the decaying subspaces are exactly all of its subspaces.
Its complement is then the unique \emph{minimal asymptotic subspace}. 
Note that this coincides with the so-called ``four-corners decomposition'' introduced in~\cite{Albert16}.

\begin{lemma} 
\label{lemma:min-asymptotic}
Let $S\subseteq\mathbb C^n$ be a subspace.
The following are equivalent: \smallskip
\begin{enumerate}[(i)]
\item \label{it:min-asymptotic} $S$ is the minimal asymptotic subspace.
\item \label{it:all-fixed} $S$ is the smallest subspace containing the support of all fixed points $\rho\in\ker L$.
\item \label{it:all-min} $S$ is the span of all minimal invariant subspaces.
\end{enumerate} \smallskip
\end{lemma}

\begin{proof} 
The equivalence of~\ref{it:all-fixed} and~\ref{it:all-min} follows from Lemma~\ref{lemma:min-collecting}.
It is easy to see that the support of a fixed point is contained in the minimal asymptotic subspace.
The converse is shown in~\cite[Prop.~15]{BNT08b}.
This proves the result.
\end{proof}

An algebra of complex matrices which is closed under taking the adjoint is called a $*$-algebra.
If $\mc A$ is a $*$-algebra, then $S$ is invariant under $\mc A$ if and only if $P_SA=AP_S$ for all $A\in\mc A$, where $P_S$ is the orthogonal projection onto $S$.

\begin{definition}
We define the \emph{relaxation $*$-algebra} $\cV_*$ and the \emph{extended relaxation $*$-algebra} $\cV_*^+$ as
$$
\cV_* = \generate{\mathds1,V_k,V_k^*:k=1,\ldots,r}{alg}\,,
\qquad
\cV_*^+ = \generate{\mathds1,H,V_k,V_k^*:k=1,\ldots,r}{alg}\,.
$$
\end{definition}
Clearly, $\cV_*$ and $\cV_*^+$ are the $*$-algebras generated by $\cV$ and $\cV^+$ respectively, and hence they are well-defined.

\begin{lemma}
The following are equivalent: \smallskip
\begin{enumerate}[(i)]
\item \label{it:no-decay} There is no decay, i.e.~$\mathbb C^n$ is the minimal asymptotic subspace.
\item \label{it:orthocompl} The lattice of collecting subspaces is orthocomplemented (cf.~footnote~\ref{footnote_orthocomplemented}) and so $\cV^+$ is a $*$-algebra.
\end{enumerate}
\end{lemma}

\begin{proof}
The proof of~\cite[Prop.~14]{BNT08b} shows that if there is no decay, the lattice of collecting subspaces is orthocomplemented.
Thus, by~\cite[Thm.~11.5.1]{Gohberg06} $\cV^+$ is a $*$-algebra.
Conversely, if the lattice of collecting subspaces is orthocomplemented, it is atomistic\footnote{A lattice with least element $0$ is atomistic if every element is a least upper bound of a set of atoms, which are the minimal non-zero elements.}, and hence by Lemma~\ref{lemma:min-asymptotic} there is no decay.
\end{proof}

\begin{lemma} \label{lemma:enclosure}
Let $S\subseteq\mathbb C^n$ be a subset, and $P_S$ the orthogonal projection onto $S$.
Then the following are equivalent: \smallskip
\begin{enumerate}[(i)]
\item \label{it:biconf} Both $S$ and $S^\perp$ are collecting.
\item \label{it:star-invariant} $S$ is an invariant subspace of $\cV_*^+$.
\item \label{it:commute} $[P_S,V_k]=[P_S,H]=0$ for $k=1,\ldots,r$.
\end{enumerate}
We call such subspaces \emph{enclosures}.
\end{lemma}

\begin{proof}
\ref{it:biconf} $\Leftrightarrow$ \ref{it:star-invariant}: This is clear since $\cV_*^+$ is the $*$-algebra generated by $\cV^+$.
\ref{it:star-invariant} $\Leftrightarrow$ \ref{it:commute}: Elementary.
\end{proof}
Enclosures imply the existence of conserved quantities and (dynamical) symmetries of the Kossakowski--Lindblad generator~\cite{BNT08b,Albert14}, but they are not necessary.

\begin{lemma}
\label{lemma:generate-matrix-algebra}
Assume that $A\in\C^{n\times n}$ is not a scalar multiple of the identity. Then there exists a Hermitian matrix $B\in\C^{n\times n}$ such that they generate the entire matrix algebra, that is, $\generate{A,B}{alg}=\C^{n\times n}$.
\end{lemma}

\begin{proof}
Let $B\in\mathbb R^{n\times n}$ is any diagonal matrix with distinct non-zero diagonal elements.
If we can find $U\in\U(n)$ unitary such that $(U^*AU)_{jk}\neq 0$ for all $j\neq k$, then \cite[Lem.~2]{Laffey92} implies that the algebra generated by $\{U^*AU,B\}$ is all of $\mathbb C^{n\times n}$;
in particular $\generate{A,UBU^*}{alg}=\C^{n\times n}$ which would conclude the proof.
Consider the function $U\mapsto (U^*AU)_{ij}$ defined on the unitary group and let $Z_{ij}$ denote the corresponding zero set. 
Since the unitary group is a real analytic manifold and since the function is real analytic, $Z_{ij}$ is either all of $\U(n)$ or it has open dense complement in $\U(n)$, cf.~\cite[Thm.~6.3.3]{Krantz02}.
If $A$ is not a multiple of the identity, none of the $Z_{ij}$ equal $\U(n)$. 
Hence the union of all $Z_{ij}$ has open dense complement, and hence there is $U$ such that $(U^*AU)_{jk}\neq 0$ for all $j\neq k$, as desired.
\end{proof}

\begin{corollary}
\label{coro:generate-matrix-algebra}
Let $A,C\in\C^{n\times n}$ such that $A$ is not a multiple of the identity. Then there exists $B\in\C^{n\times n}$ Hermitian such that $\generate{A,B+C}{alg}=\C^{n\times n}$.
\end{corollary}

\begin{proof}
Let $B,U$ be as in the proof of Lemma~\ref{lemma:generate-matrix-algebra}.
By \cite[Ch.~2~§1]{Kato80}, for all $\epsilon>0$ small enough, there exist analytic curves $S(\epsilon)$ and $B(\epsilon)$ such that $B+\epsilon U^*CU = S(\epsilon)B(\epsilon)S(\epsilon)^{-1}$, with $S(0)=\mathds1$ and $B(\epsilon)$ diagonal with distinct non-zero diagonal elements. 
Then
\begin{align*}
\generate{A,\tfrac{UBU^*}{\epsilon}+C}{alg}
&=
U\generate{U^*AU,S(\epsilon)B(\epsilon)S(\epsilon)^{-1}}{alg}U^*
\\&= 
US(\epsilon)\;\generate{S(\epsilon)^{-1}AS(\epsilon),B(\epsilon)}{alg}\;S(\epsilon)^{-1}U^*\,.
\end{align*}
Since for $\epsilon$ small enough it still holds that $(S(\epsilon)^{-1}U^*AU S(\epsilon))_{jk}\neq 0$ for all $j\neq k$, the result follows again from \cite[Lem.~2]{Laffey92}.
\end{proof}

Using Corollary~\ref{coro:generate-matrix-algebra} together with~\cite[Thm.~12]{TV09} we obtain a useful extension of Corollary~\ref{coro:stab-ham}.

\begin{lemma} \label{lemma:unique-attractive-fp}
Let $S\subseteq\mathbb C^n$ be a lazy subset for $-L$. Then there exists a Hamiltonian $H$ such that $-(\iu\ad_H+L)$ has a unique (attractive) fixed point with support $S$ if and only if $S$ is not an enclosure. 
\end{lemma}

\begin{proof}
In \cite[Thm.~12]{TV09} it is shown how to choose $H$ such that $S$ is attractive.
We are still free to choose $H$ on $S$ itself.
So using Corollary~\ref{coro:generate-matrix-algebra} we can make sure that $S$ is minimal collecting and hence by Lemma~\ref{lemma:min-collecting} it supports a unique fixed point.
\end{proof}

\bibliographystyle{siamplain}
\bibliography{../Bibliography/general.bib}

\end{document}